\documentclass[journal,twoside,web]{ieeecolor}
\makeatletter
\def\BibTeX{{\rm B\kern-.05em{\sc i\kern-.025em b}\kern-.08em
    T\kern-.1667em\lower.7ex\hbox{E}\kern-.125emX}}
\makeatother

\usepackage{generic}
\usepackage{cite}
\usepackage{amsmath,amssymb,amsfonts,amsthm}
\usepackage{algorithmic}
\usepackage{textcomp}
\usepackage{bm}
\usepackage[justification=centering,font=small]{caption}
\usepackage{pict2e}
\usepackage{dsfont}
\usepackage{mathtools}	
\usepackage{graphicx}
\usepackage{psfrag}
\graphicspath{{./graphics/}{./figs/}}					
\mathtoolsset{showonlyrefs,showmanualtags}		
\usepackage{subfigure}
\usepackage{siunitx}
\usepackage{comment}
\usepackage{tikz}
\usetikzlibrary{automata,shapes,positioning,calc,math}
\usepackage{pgfplots}
\usepgfplotslibrary{fillbetween}

\definecolor{janniks_1A_blau}{RGB}{0,114,190}
\definecolor{janniks_1A_rot}{RGB}{218,83,25}
\definecolor{janniks_1A_grun}{RGB}{119,173,48}
\definecolor{richards_1A_blau}{RGB}{197,217,241}
\definecolor{richards_1A_rot}{RGB}{192,0,0}
\definecolor{richards_1A_grun}{RGB}{146,208,80}

\newcommand{\CNET}{CNET}
\newcommand{\CSYS}{CSYS}

\newcommand{\T}{\intercal}

\newcommand{\Hc}{{H}}
\newcommand{\Hj}{{\hat{H}}}
\newcommand{\HR}{{N}}

\newcommand{\B}[1]{\mathbb{B} \left[ #1 \right]}

\newcommand{\CE}[2]{\mathbb{E} \mathopen{} \left[ \mathclose{} #1 \, \middle| \, #2 \right]}
\newcommand{\tr}[1]{\tilde{#1}}
\newcommand{\e}[1]{\bar{#1}}

\newcommand{\diag}[2]{ \operatorname{diag}_{#1} \left\{ #2 \right\} }

\newcommand{\s}[2][]{^{#2^{#1}}}
\newcommand{\norm}[1]{\left\lVert#1\right\rVert}
\newcommand{\setdef}[2]{\left\{ #1 \left\vert #2 \right.\right\}}
\newcommand{\tbf}[2][]{\mathbf{#1\tilde{#2}}}
\newcommand{\tbfi}[2][]{\mathbf{#1\tilde{#2}}\s{i}}
\newenvironment{sbm}{\setlength\arraycolsep{2pt}\bmatrix}{\endbmatrix}
\newenvironment{cases2}{\left. \begin{cases}}{\end{cases} \right\}}
\newenvironment{salign}{\thinmuskip=1mu\medmuskip=2mu\thickmuskip=3mu\align}{\endalign}
\newcommand{\figref}[1]{Fig.~\ref{#1}}
\newcommand{\secref}[1]{Section~\ref{#1}}
\newcommand{\Qxi}{{Q_x\s{i}}}
\newcommand{\QxiT}{{Q_x\s[\T]{i}}}
\newcommand{\Qui}{{Q_u\s{i}}}
\newcommand{\QuiT}{{Q_u\s[\T]{i}}}
\newcommand{\dum}{z}            
\newcommand{\aoi}{a}
\newcommand{\Ni}{\mathcal{N}\s{i}}
\newcommand{\nc}{n}                 
\newcommand{\zu}{\null\\}           

\newcommand{\kp}[1]{}
\DeclareFontFamily{U}{mathb}{}
\DeclareFontShape{U}{mathb}{m}{n}{
	<-5.5> mathb5
	<5.5-6.5> mathb6
	<6.5-7.5> mathb7
	<7.5-8.5> mathb8
	<8.5-9.5> mathb9
	<9.5-11.5> mathb10
	<11.5-> mathbb12
}{}
\theoremstyle{plain}
\newtheorem{theorem}{Theorem}

\newtheorem{assume}[theorem]{Assumption}
\newtheorem{definition}[theorem]{Definition}

\begin{document}
\title{Using AoI Forecasts in Communicating and Robust Distributed Model-Predictive Control}

\author{J.~Hahn$^1$,
	R.~Schoeffauer$^2$,
	G.~Wunder$^2$,
	O.~Stursberg$^1$. 
	\thanks{Submitted to review on \today.}
	\thanks{This work is financially supported in part by the German Research Foundation (DFG) within the research priority program SPP1914 -- \textit{Cyberphysical Networking}.}
	\thanks{$^{1}$Control and System Theory, Dept. of Electrical Engineering and Computer Science, University of Kassel, Germany; \{\tt{jhahn, stursberg\}@uni-kassel.de}}%
	\thanks{$^2$Heisenberg CIT Group, Free University of Berlin, Dept. of Mathematics and Computer Science, Germany; \tt{richard.schoeffauer@fu-berlin.de}}
}


\maketitle

\begin{abstract}
	In order to enhance the performance of cyber-physical systems, this paper proposes the integrated design of distributed controllers for distributed plants and the control of the communication network. Conventional design methods use static interfaces between both entities and therefore rely on worst-case estimations of communication delay, often leading to conservative behavior of the overall system. By contrast, the present approach establishes a robust distributed model-predictive control scheme, in which the local subsystem controllers operate under the assumption of a variable communication schedule that is predicted by a network controller. Using appropriate models for the communication network, the network controller applies a  predictive network policy for scheduling the communication among the subsystem controllers across the network. Given the resulting time-varying predictions of the age of information, the paper shows under which conditions the subsystem  controllers can robustly stabilize the distributed system. To illustrate the approach, the paper also reports on the application to a vehicle platooning scenario.
\end{abstract}


\IEEEpeerreviewmaketitle

\section{Introduction}
\label{sec:intro}

Current technological advances in communication technology have lead to systems in which networks connect more and more locally controlled and autonomously operating devices.
Such systems have an impact across a large number of applications,
including networked automobile and traffic systems, smart energy grids, and the next generation of manufacturing plants (\textit{industry 4.0}).
Typically referred to as \emph{cyber-physical systems}, these systems are composed of physical components, digital and computational nodes, as well as the interconnecting communication infrastructure \cite{lee2008cyber}. While traditional engineering concepts follow \textit{divide-and-conquer} principles for separated and largely decoupled design of these components, requirements of high performance, reliability, as well as online and autonomous reconfiguration call for integrated design in which inter-dependencies are carefully taken into account \cite{5995279}.

This paper proposes a new approach of the latter type, tailored to the specific case of combining a wireless communication network with a distributed plant in which subsystems are locally controlled by model predictive controllers (MPC). The controllers aim at establishing cooperation with respect to a common cost functional formulated for the distributed plant,
thus requiring to exchange data between the controllers across a centrally organized communication network with possibly time-varying properties of connectivity, reliability, and latency.

In order to let the controllers adapt to such properties autonomously and online, the use of model predictive controllers is a straightforward choice: they do not only allow to consider the predicted behavior of other controlled entities and constraints for states, inputs, and outputs, but also imperfections in the communication.
In fact, model predictive control without consideration of communication defects has reached a state of considerable maturity, including variants for nonlinear dynamics \cite{GP17,RMM94}, systems with uncertainties \cite{LK00,MSR05,BB12,CF13}, for fast computations \cite{WB10,GC12,ZJM11}, and distributed settings \cite{SV+10,CS+13,CJ+02}.
 
With respect to versions of distributed MPC taking network imperfections into account, solutions have been proposed in \cite{LX+07,
	Gr+14,gross2014distributed}. Common lines in these studies are, however, that network delays are either assumed to be negligibly small compared to the 
dominant plant dynamics \cite{shi2017distributed,
	blasi2018distributed,tangirala2018analysis}, or that an upper bound of the delay (commonly named as \textit{worst case delay}) is assumed to be known \cite{jia2014network,yamchi2017distributed}. Own work in this direction has aimed at devising robust MPC strategies to compensate for the maximum delays \cite{gross2014distributed,GS11}, or to use schemes of event-based communication \cite{GS16}. 

However, explicitly accounting for the network defects by use of a worst case delay within robust control schemes is, most of the time, overly conservative: It can often be observed, that considerable delays do only occur infrequently and accumulated in certain phases, while for longer duration the delay is negligible \cite{ploplys2004closed}.
Consequently, the subsystem controllers (network agents) typically hold much fresher data than what would be expected under the worst case delay.
In addition, a time-varying communication schedule leads to non-uniformly distributed instances of information reception and therefore lends itself to a description via the so-called \textit{age of information} (AoI) metric that measures the time elapsed between generation and reception of information. 
This motivates to develop methods that can obtain and make use of the expected AoI, thus circumventing the static interface between communication and control that a worst case delay typically amounts to.
Assuming that a basic model for the link quality can be obtained (e.g. via machine learning techniques \cite{wang2017, qiu2018}), this paper proposes a model predictive network controller to handle both the management of transmissions as well as the prediction of future AoI, a strategy that is reminiscent of our prior work \cite{schoeffauer2018a,schoeffauer2018b}.
In contrast to these works this paper assumes a flat topology, meaning that there exists a link between each pair of network agents, and hence routing is no longer a problem that the network controller needs to solve. Such a topology is commonly encountered in related literature, where it stands to develop an optimal scheduler in order to minimize the overall age of data in machine-to-machine communication scenarios \cite{Sinha2019, He2018, Hsu2020}.

Recent work in \cite{hahn2018distributed} has sketched the idea of using predictive controllers for both the minimization of network delays, and the control of distributed plants under consideration of predicted AoI. That paper, however, neither detailed the network control scheme, nor the interface between network and subsystem controllers, nor the stability of the overall scheme.
In contrast, the present paper proposes a novel predictive control scheme for the communication network and defines a mathematical interface through which the subsystem controllers can make use of the resulting forecasts of the AoI.
Furthermore, it shows how these forecasts can be used to enhance control performance of the subsystem controllers, and under which conditions robust stability is ensured.

The paper is organized as follows:
\secref{sec:architecture} introduces the general system architecture.
\secref{sec:network} presents the design of a predictive network control scheme that generates delay forecasts.
\secref{sec:controlsystem} presents the design of a distributed predictive control scheme that makes use of the delay forecasts and proofs its stability.
\secref{sec:simulation} illustrates the performance gain, when employing the proposed methods to a vehicle platoon scenario.


\section{Set-Up and Notation}
\label{sec:architecture}

\paragraph*{Set-up}
This paper considers cyber-physical systems composed of two main parts, a distributed control system (\CSYS) comprising a set of locally controlled subsystems, and a communication network (\CNET) over which the local controllers of the subsystems in \CSYS\ can communicate, see \figref{fig::general_system_model}. 
The dynamics of the subsystems are assumed to be decoupled in this setting, i.e. the state of one subsystem may not directly affect the dynamics of another, while different subsystems can impose constraints onto each other, and the behavior of one controlled subsystem may affect the control goal of another. An example in which such dependencies are practically relevant is that of a platoon of autonomous vehicles, as elaborated on in Sec.~\ref{sec:simulation}. It is further assumed that the local controllers are not able to measure the state of an interacting subsystem, thus information on neighbors can solely be obtained by exchange of information through the \CNET.

In the latter, each subsystem controller acts as an agent in the communication network that requests data from and provides data to all other agents. Here, "data"  refers to state information of the subsystems. The agents are connected via links which exhibit individual, time-depending behavior with respect to transmission quality. A centralized network controller manages a schedule, determining when which agent is allowed to broadcast its data to all other agents. In contrast to conventional network control approaches, this paper models the network controller as an MPC, enabling the generation of forecasts of future data transmission. In other words, it becomes possible to inform an agent of when, in terms of a bounded time interval, it will receive data from another agent. It is assumed that these forecasts are communicated to the agents with only a negligible overhead of transmission load through the network.

Accordingly, as indicated in \figref{fig::general_system_model}, the \CNET\ architecture foresees two layers of communication: on one layer the agents broadcast their data to one another; on the other layer, the network controller informs the agents of both, the broadcast schedule (via control orders) and the forecasts.

\begin{figure}[t!]
	\centering
	\begingroup%
	\makeatletter%
	\newcommand*\fsize{\dimexpr\f@size pt\relax}%
	\newcommand*\lineheight[1]{\fontsize{\fsize}{#1\fsize}\selectfont}%
	\ifx\svgwidth\undefined%
	\setlength{\unitlength}{227.33857403bp}%
	\ifx\svgscale\undefined%
	\relax%
	\else%
	\setlength{\unitlength}{\unitlength * \real{\svgscale}}%
	\fi%
	\else%
	\setlength{\unitlength}{\svgwidth}%
	\fi%
	\global\let\svgwidth\undefined%
	\global\let\svgscale\undefined%
	\makeatother%
	\begin{picture}(1,0.78802996)%
		\lineheight{1}%
		\setlength\tabcolsep{0pt}%
		\put(0,0){\includegraphics[width=\unitlength,page=1]{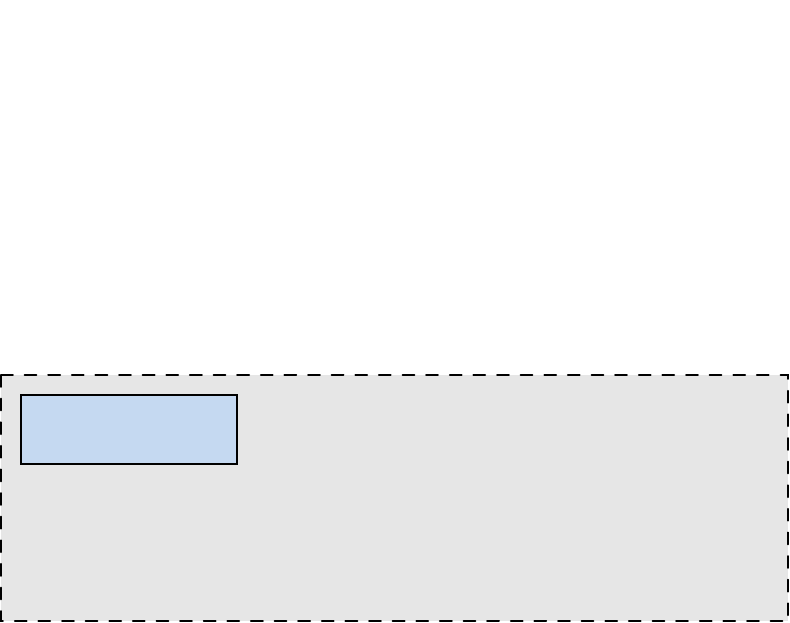}}%
		\put(0.16232196,0.22968467){\makebox(0,0)[t]{\lineheight{1.25}\smash{\begin{tabular}[t]{c}$C_1$\end{tabular}}}}%
		\put(0,0){\includegraphics[width=\unitlength,page=2]{set_up.pdf}}%
		\put(0.16201754,0.12){\makebox(0,0)[t]{\lineheight{1.25}\smash{\begin{tabular}[t]{c}$P_1$\end{tabular}}}}%
		\put(0,0){\includegraphics[width=\unitlength,page=3]{set_up.pdf}}%
		\put(0.49890788,0.22959867){\makebox(0,0)[t]{\lineheight{1.25}\smash{\begin{tabular}[t]{c}$C_2$\end{tabular}}}}%
		\put(0,0){\includegraphics[width=\unitlength,page=4]{set_up.pdf}}%
		\put(0.4986126,0.12){\makebox(0,0)[t]{\lineheight{1.25}\smash{\begin{tabular}[t]{c}$P_2$\end{tabular}}}}%
		\put(0,0){\includegraphics[width=\unitlength,page=5]{set_up.pdf}}%
		\put(0.83550294,0.22963026){\makebox(0,0)[t]{\lineheight{1.25}\smash{\begin{tabular}[t]{c}$C_n$\end{tabular}}}}%
		\put(0,0){\includegraphics[width=\unitlength,page=6]{set_up.pdf}}%
		\put(0.83527092,0.12){\makebox(0,0)[t]{\lineheight{1.25}\smash{\begin{tabular}[t]{c}$P_n$\end{tabular}}}}%
		\put(0,0){\includegraphics[width=\unitlength,page=7]{set_up.pdf}}%
		\put(0.49866521,0.64524049){\makebox(0,0)[t]{\lineheight{1.25}\smash{\begin{tabular}[t]{c}Network Controller\end{tabular}}}}%
		\put(0,0){\includegraphics[width=\unitlength,page=8]{set_up.pdf}}%
		\put(0.19,0.55){\color[rgb]{1,0,0}\makebox(0,0)[lt]{\lineheight{1.25}\smash{\begin{tabular}[t]{l}
						Network data
		\end{tabular}}}}%
		\put(0.57,0.43){\color[rgb]{0,0,1}\rotatebox{13.213929}{\makebox(0,0)[lt]{\lineheight{1.25}\smash{\begin{tabular}[t]{c}
							CSYS \\[-1mm] data
		\end{tabular}}}}}%
		\put(0.01371571,0.73690777){\makebox(0,0)[lt]{\lineheight{1.25}\smash{\begin{tabular}[t]{l}CNET\end{tabular}}}}%
		\put(0.98628433,0.02618455){\makebox(0,0)[rt]{\lineheight{1.25}\smash{\begin{tabular}[t]{r}CSYS\end{tabular}}}}%
	\end{picture}%
	\endgroup%
	\caption{Structure of the cyber-physical system, with local subsystem controllers $C_i$ and plant subsystems $P_i$. 
	}
	\label{fig::general_system_model}
\end{figure}

\paragraph*{Notation}
\CSYS\ and \CNET\ both operate on discrete-time domains with underlying time-steps. 
Any value of a discrete-time signal $\dum$ at time $t_0+k\cdot \Delta t$ is denoted by $\dum_{k}$ with index $k\in \mathbb{N}$ and a constant interval $\Delta t \in \mathbb{R}^+$. A value $\dum_{k+l}$ predicted in time $k$ is indicated by $\dum_{k+l|k}$, where $\dum_{k|k} = \dum_k$. A complete predicted trajectory over a horizon of length $H$ is denoted by $\tilde{\dum}_k^\T= [ \dum_{k|k} ^\T,\ \dum_{k+1|k}^\T,\ \dots,\ \dum_{k+H|k}^\T ]^\T$.
With slight abuse of notation, in some cases the first or last entry in the trajectory is omitted.


The symbol $\mathbf{\dum}_k$ refers to a stacked column vector of signals from different subsystems, and a trajectory of a stacked vector $\mathbf{\dum}_k$ is denoted by $\mathbf{\tilde{\dum}}_k$.

To refer to polytopic constraints for any element of \mbox{$\mathbf{\tilde{\dum}}_k$: $\mathcal{
		\expandafter\MakeUppercase\expandafter{\dum}}=\big\{ \mathbf{\tilde{\dum}}_k \in \mathbb{R}^{\tbf{n}_\dum}\big\vert \tbf{C}_\dum \cdot \tbf{\dum}_k \leq \tbf{b}_\dum\big\}$} the pair  $(\tbf{C}_\dum,\tbf{b}_\dum)$ is used.

Matrices ${0}_{n,m}$ and ${1}_{n,m}$ denote $n \times m$-matrices of zeros, or ones respectively, while a column vector is simpler written as ${0}_n$. For brevity, $\mathbf{0}$ is sometimes used to denote a zero matrix, if the dimensions are clear from the context.


Let an index set $\mathcal{N}=\{1,2,\ldots,n\}\subset\mathbb{N}$ refer to a set of $n$ subsystems. Then, a column vector $\dum\s{i}$, a matrix $A\s{i}$, and a set $\mathcal{A}\s{i}$ indicate variables defined for the subsystem with index $i\in\mathcal{N}$. Furthermore, for the example of \mbox{$\mathcal{N}=\{1,2,3\}$}, the notation $\mbox{diag}\left(\begin{sbm}A\s{j}\end{sbm}, {j\in\mathcal{N}}\right)$ is equivalent to $\mbox{diag}(A\s{1},A\s{2},A\s{3})$, and $\mbox{prod}\left(\mathcal{A}\s{j}, j\in\mathcal{N}\right)$ defines the Cartesian product $\mathcal{A}\s{1}\times\mathcal{A}\s{2}\times\mathcal{A}\s{3}$. In addition, $[\dum\s{j}]_{j\in\mathcal{N}}$ defines the stacked vector $[\dum\s[\T]{1},\ \dum\s[\T]{2},\ \dum\s[\T]{3}]^\T$. 
\section{Communication Network}
\label{sec:network}

The \CNET\ is modeled as a discrete-time packet-based system with $n$ agents (corresponding to $n$ subsystem controllers) and erroneous (wireless) transmissions.
It is assumed that in each time-step each agent requires data from all other agents, and that likewise in each time-step each agent provides a new batch of data that can potentially be broadcast to all other agents. The network resources are assumed to be  limited, and thus agents typically have to work with out-dated data until new information is received. The age of information is pivotal for the agents' performance.
\begin{definition}
	Suppose each batch of data gets a time-stamp when being generated by its agent.
	The age of information $a_k^{ij} \in \mathbb{N}$ is the difference between the current time-step $k$ and the time stamp of the latest batch of data agent $i$ received from agent $j$.
\end{definition} 

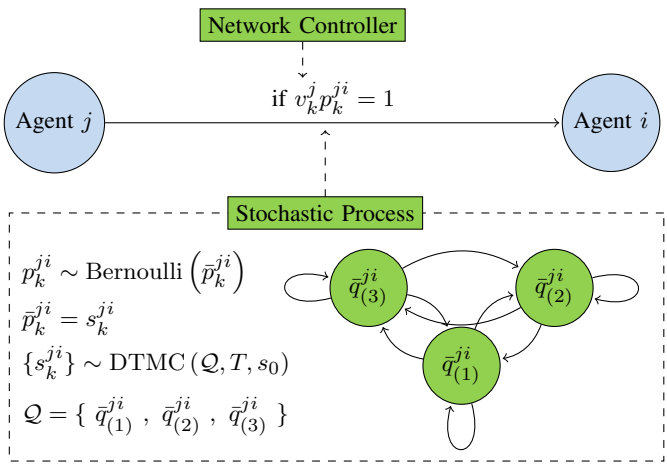
\begin{figure}[b!]
	\centering
	\small
	\begin{tikzpicture}[xscale=0.1,yscale=0.1]
		\tikzset{state/.style={draw, fill=richards_1A_blau}};
		
		\node[state, circle] (Aj) at (-10,0) {Agent $j$};
		\node[state, circle] (Ai) at (64,0)    {Agent $i$};
		\draw[every loop] (Aj) edge[auto] node {if $v^j_k p^{ji}_k = 1$} (Ai);
		
		\coordinate (v) at (23,3);
		\node[rectangle, draw, fill=richards_1A_grun] (NC) at ($ (v) + (0,10) $) {Network Controller};
		
		\draw[->, dashed] (NC) -- ($ (v) + (0,3) $);
		
		\coordinate (c) at ($ (v) + (3,0) $);
		\coordinate (sp) at ($ (c) - (0,15) $);
		
		\draw[->,dashed] (sp) -- ($ (c) - (0,4) $);
		
		\node[anchor = north] (p) at ($ (sp) + (-16,-4) $)
		{
			$ \begin{aligned}
				&p_k^{ji} \sim \operatorname{Bernoulli} \left( \e{p}_k^{ji} \right)
				\hphantom{ \ , \quad \e{p}_k^{ji} = s_k^{ji}}
				\\
				&\e{p}_k^{ji} = s_k^{ji}
				\\
				&\{ s_k^{ji} \} \sim \operatorname{DTMC} \left( \mathcal{Q} , T, s_0 \right)
				\vphantom{\left( \e{p}_k^{ji} \right)}
				\\
				&\mathcal{Q} =
				\{ \
				\e{q}_{(1)}^{ji} \ , \
				\e{q}_{(2)}^{ji} \ , \ \e{q}_{(3)}^{ji} \ \}
				\vphantom{\left( \e{p}_k^{ji} \right)}
			\end{aligned} $
		};
		
		\node[state, circle, fill=richards_1A_grun, above right = -20mm and 5mm of p] (s1) {$\e{q}^{ji}_{(1)}$};
		\node[state, circle, fill=richards_1A_grun, above right = 3mm and 5mm of s1] (s2)  {$\e{q}^{ji}_{(2)}$};
		\node[state, circle, fill=richards_1A_grun, above left = 3mm and 5mm of s1] (s3)  {$\e{q}^{ji}_{(3)}$};
		
		\draw[every loop]
		(s1) edge[loop below] node {} (s1)
		(s1) edge[bend left] node {} (s2)
		(s1) edge[bend left] node {} (s3)
		
		(s2) edge[bend left] node {} (s1)
		(s2) edge[loop right] node {} (s2)
		(s2) edge[bend left] node {} (s3)
		
		(s3) edge[bend left] node {} (s1)
		(s3) edge[bend left] node {} (s2)
		(s3) edge[loop left] node {} (s3);
		
		\coordinate (t1) at ($ (sp) + (45,0) $);
		\coordinate (t2) at ($ (sp) + (-42,0) $);
		\coordinate (t3) at ($ (sp) + (-42,-33) $);
		\coordinate (t4) at ($ (sp) + (45,-33) $);
		\draw[dashed] (t1) -- (t2) -- (t3) -- (t4) -- (t1);

		\node[rectangle, draw, fill=richards_1A_grun] at (sp) {Stochastic Process};
		
	\end{tikzpicture}
	\caption{Conditions for successful transmission from agent~$j$ to agent~$i$: $v_k^j p_k^{ji}$ has to be $1$; the DTMC stands exemplarily. }
	\label{fig::link}
\end{figure}

If, in the current time-step $k$, agent $i$ successfully receives data from agent $j$, $\aoi^{ij}_k$ is reset to 1 (since in time-discrete models, transmission and computation is assumed to take up an entire time-step). Otherwise $\aoi^{ij}_k$ is increased by 1.
However, successfully receiving data from agent $j$ can only occur, if (i) the network controller does schedule agent $j$ to broadcast its data, via setting the control variable $v_k^j \in \{0,1\}$ to $1$, and if (ii) the data does not get lost due to erroneous transmissions (see \figref{fig::link}). The second part is expressed by the stochastic variable $p^{ji}_k \in \{0,1\}$ ($0$ corresponding to failure) such that data is successfully received if $v_k^j \cdot p^{ji}_k = 1$. This results in the following evolution for the AoI:
\begin{equation}
	\label{eq::com::aoi_evolution}
	\aoi^{ij}_{k+1} 
	= 1 + \aoi^{ij}_k \left( 1 - v^j_k p^{ji}_k \right)
	= 1 + \aoi^{ij}_k \left( 1 - p^{ji}_k \right)^{v^j_k }
	.
\end{equation}

The control variables $v_k^j$ are collected in the binary control vector $v_k = [v_k^1,\dots v_k^n]^\intercal \in \{0,1\}^n$. Usually, e.g. due to interference properties, only certain agents are allowed to be engaged in broadcasting at the same time, and hence only certain realizations of $v_k$ are admissible. The admissible set of control vectors will be denoted by $\mathcal{V}$ such that $v_k \in \mathcal{V}$.

For ease of notation, the upper indices of the random variable $p_k^{ji}$ are omitted for the moment: the process $\{p_k\}$ is governed by a Bernoulli process and a discrete-time Markov chain (DTMC) as shown in the bottom of \figref{fig::link}. The first one allows for consideration of unpredictable short-term drops in channel quality. In particular, each $p_k$ is Bernoulli distributed (i.e. either $0$ or $1$) with success-parameter $\e{p}_k \in [0,1]$ (note that the parameter is indeed time-variant). Opposed to that, the DTMC allows for consideration of a partially predictable, long-term behavior by 
dictating the time behavior of the parameter $\e{p}_k$.
To that end, let 
$\mathcal{Q} = \left\{ \e{q}_{(1)}, \dots \e{q}_{(m)} \right\}$
be a finite set of values that $\e{p}_k$ can take (the indices of the elements are set in parentheses to distinguish them from the time-step). Moreover, define
$\{s_k\} \sim \operatorname{DTMC}\left(   \mathcal{Q} , T, s_0 \right)$ with transition matrix $T$ and initial state $s_0$ such that in each time-step $k$ it holds that $\e{p}_k = s_k$. For completeness, let the row vector $\sigma_k$ denote the distribution to $s_k$. The following common assumption is made:
\begin{assume}
	The quantities $\mathcal{Q}$ and $T$ of the DTMC $\{s_k\}$ are known (e.g. as a result of machine learning techniques), and the state $s_k$ is observable in time-step $k$.
\end{assume}

Note that under this model, it is possible to calculate the following term (needed later in \eqref{eq::com::future_aoi}) for some arbitrary time-steps $k_0 < k_1 < \dots k_l$ in increasing order:
	\begin{align}
	    \CE{ \prod_{i=1}^l p_{k_i} }{ \sigma_{k_0}}
	=
	\CE{ \prod_{i=1}^l \e{p}_{k_i} }{ \sigma_{k_0}}
	\\
	=
	\sigma_{k_0} \prod_{i = 1}^l \left( T^{k_i-k_{i-1}} \Delta_\mathcal{Q} \right) \mathbf{1}
	=
	\sigma_{k_0} T^{k_l} \Delta_{\mathcal{Q}}^{l} \mathbf{1}
	,
	\end{align}
where $\Delta_\mathcal{Q}$ is the diagonal matrix whose entries are the elements of $\mathcal{Q}$ in the given order.

An obvious control objective consists of minimizing the AoI for the \textit{immediate next} time-step (a common approach in current network control strategies). Opposed to that, our strategy aims to minimize the AoI over a certain number of time-steps, the prediction horizon $\HR$. Without loss of generality, assume that the current time-step is $0$. Then, using the recursive form \eqref{eq::com::aoi_evolution}, the explicit expression for the AoI in time-step $k$ becomes
\begin{align}
    	\label{eq::com::future_aoi}
	\aoi_{k}^{ij} &=
	1 + \left( 1-v_{k-1}^{j} p_{k-1}^{ji} \right)
	\\
	& \hspace{6.5mm} 
	+ \left( 1-v_{k-1}^{j} p_{k-1}^{ji} \right)
	\cdot
	\left( 1-v_{k-2}^{j} p_{k-2}^{ji} \right)
	+ \dots
	\\
	& \hspace{6.5mm} 
	+
	\left( 1-v_{k-1}^{j} p_{k-1}^{ji} \right)
	\cdot \ \dots \ \cdot
	\left( 1-v_{0}^{j} p_{0}^{ji} \right) a_0^{ij}
	\\
	&=
	k +
	\sum_{\mathcal{I} \in P( \mathbb{N}_{k-1} )}
	\left( -1 \right)^{|\mathcal{I}|} \left( \aoi_0 + \min \{ \mathcal{I} \} \right)
	\prod_{l \in \mathcal{I} } v_{l}^j p_{l}^{ji}
	.
\end{align}
Here, $P( \mathbb{N}_{k-1} )$ is the power set of the set $\mathbb{N}_{k-1}$ (the natural numbers from $0$ up to $k-1$) which stems from the products of $(1 - v_l p_l)$. It is $\min \{ \emptyset \} = \max \{ \emptyset \} = 0$ and $\prod_{\emptyset} = 1$. 

Summing up \eqref{eq::com::future_aoi} for $k = 1,\dots N$ and taking the expectation yields the suitable objective $J^{ij}$ (the accumulated expected AoI values over the next $\HR$ time-steps):
\begin{align}
	\label{eq::com::part_original_objective}
	J^{ij}(&v_{0|0},\dots v_{\HR-1|0}) := \CE{ \sum_{k=1}^\HR \aoi_{k}^{ij} }{a_0^{ij} , \sigma_0} 
	\\
	=&
	\frac{\HR^2+\HR}{2} +
	\sum_{\mathcal{I} \in P( \mathbb{N}_{\HR-1} )}
	\left( \HR - \max \{ \mathcal{I} \} \right)
	\\
	\cdot 
	&\left( -1 \right)^{|\mathcal{I}|}
	\left( \aoi_0^{ij} + \min \{ \mathcal{I} \} \right)
	\left( \prod_{l \in \mathcal{I} } v_{l|0}^{j} \right) \CE{ \prod_{l \in \mathcal{I} } p_{l}^{ji} }{\sigma_0}
	.
\end{align}
Note that there is no need to consider the sum of squares of $a_{k}^{ij}$, because every AoI is potentially reset to $1$, no matter its value. Hence, larger values of the AoI will automatically be more prone to minimization than smaller ones, even in this linear formulation.

Note that \eqref{eq::com::part_original_objective} merely considers the AoI for data of agent $i$ from agent $j$. The actual control objective however has to consider all AoI values in the entire system, and hence becomes:
\begin{gather}
	\label{eq::com::single_objective}
	\min_{ v_{0|0}, \dots v_{\HR-1|0} \in \mathcal{V} } \
	\sum_{i = 1}^n \sum_{j = 1}^n
	w^{ij} J^{ij}( v_{0|0}, \dots v_{\HR-1|0} )
	,
\end{gather}
where $w^{ij} \in \mathbb{R}_+$ are weights to balance for more or less important data. E.g., if agent $i$ has no use for data from agent $j$, then $w^{ij} = 0$. Eventually, these weights could also be signaled from the agents to the network controller, allowing for time-variant weights.

Due to the control vector $v_k$ being binary, this is a combinatorial problem with non-linear objective function. The amount of feasible solutions is given by $|\mathcal{V}|^\HR$ where $|\mathcal{V}|$ denotes the cardinality of $\mathcal{V}$. Furthermore, the amount of summands (over the power set in \eqref{eq::com::part_original_objective}) that need be evaluated in order to obtain the value of $J^{ij}$ for a single realization of $v_{0|0},\dots v_{\HR-1|0}$ grows exponentially in $\HR$ as well. Hence, the problem is most certainly not suited for practical implementation.

To alleviate the computational burden, two heuristic relaxations to the optimization problem \eqref{eq::com::single_objective} are introduced that drastically reduce the effort of finding a good solution (evaluated in extensive simulations). First, the branching of the paths of the DTMC is replaced by a mean distribution. In particular, this means that $\CE{{p}_{k_1}^{ji} \cdot {p}_{k_2}^{ji}}{\sigma_0}$ is approximated by $\CE{{p}_{k_1}^{ji}}{\sigma_0} \cdot \CE{{p}_{k_2}^{ji}}{\sigma_0}$.
Using the substitution
\begin{equation}
	\phi_k^{ji}(v_k) := \CE{ 1-v_k^j p_k^{ji} }{\sigma_0}
	\in
	\{ 1 , 1 - \sigma_0 T^k \Delta_\mathcal{Q} \mathbf{1} \}
\end{equation}
this makes it possible to state \eqref{eq::com::part_original_objective} in a simplified form:
\begin{multline}
	I^{ij}(v_{0|0},\dots v_{\HR-1|0})
	:=
	\CE{ \sum_{k=1}^\HR \aoi_{k}^{ij} }{a_0^{ij} , \sigma_0}_\text{relaxed}
	\\
	=
	\HR + \sum_{k=1}^\HR \left(
	\sum_{l=1}^{k-1} \prod_{m=1}^l \phi_{m}^{ji}(v_{m|0}) +
	\aoi_0^{ij} \prod_{m=0}^{k-1} \phi_{m}^{ji}(v_{m|0}) 
	\right)
	.
\end{multline}
Compared to the strict formulation, only $\frac{\HR^2}{2}$ terms need to be evaluated in order to obtain $I^{ij}$ for a single realization of $v_{0|0},\dots v_{\HR-1|0}$.

As a second relaxation, problem \eqref{eq::com::single_objective} is separated into $\HR$ consecutive minimization problems:
\begin{gather}
	\label{eq::com::original_objective}
	\min_{
		\substack{
			v_{k_\HR|0} \in \mathcal{V} \\
			k_\HR \in \mathbb{N}_{\HR-1}\setminus \{k_{1}, \dots k_{\HR-1}\} }
	}
	\dots
	\min_{
		\substack{
			v_{k_2|0} \in \mathcal{V} \\
			k_2 \in \mathbb{N}_{\HR-1} \setminus \{k_1\} }
	}
	\
	\min_{
		\substack{
			v_{k_1|0} \in \mathcal{V} \\
			k_1 \in \mathbb{N}_{\HR-1} }
	}
	\\
	\sum_{i = 1}^n \sum_{j = 1}^n
	w^{ij} I^{ij}( v_{0|0}, \dots v_{\HR-1|0} )
	.
\end{gather}
In each minimization, the most promising realization of $v_k$ in the most promising time-step $k$ is fixed and applied to the objective. However, in any subsequent minimization, this time-step then becomes unavailable and thus the relative feasible set for each subsequent minimization shrinks. This scheme results in $(\HR + (\HR-1) + \ldots +1) \cdot|\mathcal{V}| = \frac{\HR^2+\HR}{2} \cdot |\mathcal{V}|$ evaluations of the objective before finding the solution. Both relaxations therefore drastically reduce the complexity of the problem and allow for a fast calculation of a suboptimal solution. Note that choosing $\HR=1$ still leads to the \textit{exact} minimization over all \textit{immediate} AoI values in the entire system.

In summary, the network controller operates as an MPC, solving \eqref{eq::com::original_objective} in every time-step $k$, with \eqref{eq::com::original_objective} being denoted relative to time-step $k$. Let us revoke this relative notation and assume the current time-step to be $k$ again.
Considering more than just the AoI of the immediate next time-step in the control objective, naturally improves the control performance. However, a solution of \eqref{eq::com::original_objective} now also enables the network controller to predict when new data is probably arriving at the agents. Hereafter, these predictions are called "forecasts" and explain their derivation in the following. Forecasts are defined as trajectories of predicted AoI and are denoted by:
\begin{align}
	\label{eq:forecasts}
	\tr{\aoi}_k^{ij} = [\aoi_{k+1|k}^{ij},\aoi_{k+2|k}^{ij},\dots a_{k+\HR|k}^{ij}]^\T
	. 
\end{align}

Given a current network control trajectory $\tr{v}_k = [v_{k|k}^\T ,\dots v_{k+\HR-1|k}^\T]^\T$, forecasts are generated by the network controller based on the accumulated transmission failure probability by time-step $h$, which is:
\begin{equation}
	\prod_{l=0}^{h-1} \CE{
		\left( 1 - p_{k+l}^{ji} \right)^{v_{k+l}^j}
	}{\sigma_k}
	.
\end{equation}
As soon as this probability falls beneath a certain threshold $\tau$, the network controller assumes that transmission did succeed at least once by the end of that time-step. Let $k_\tau^{ij}$ be the time-step in which this happens and let $k_f^{ij}$ be the first time-step in which transmission was attempted:
\begin{gather}
	k_\tau^{ij} := \min_{h}
	\left\{
	h:
	\prod_{l=0}^{h-1} \CE{
		\left( 1 - p_{k+l}^{ji} \right)^{v_{k+l}^j}
	}{\sigma_k} < \tau
	\right\},
	\\
	k_f^{ij} := \min_{h}
	\left\{
	h:
	\prod_{l=0}^{h-1} \CE{
		\left( 1 - p_{k+l}^{ji} \right)^{v_{k+l}^j}
	}{\sigma_k} < 1
	\right\}
	.
\end{gather}
Then, as depicted in \figref{fig::com::forecast}, the forecast $\tr{a}_k^{ij}$ becomes:
\begin{salign}
	\label{eq::com::aoi}
	\tr{\aoi}_k^{ij} =
	\begin{bmatrix}
		\aoi_{k+1|k}^{ij} \\
		\vdots \\
		\aoi_{k+k_\tau^{ij}-1|k}^{ij} \\
		\aoi_{k+k_\tau^{ij} \pm 0|k}^{ij} \\
		\aoi_{k+k_\tau^{ij}+1|k}^{ij} \\
		\vdots
	\end{bmatrix}
	=
	\begin{bmatrix}
		\aoi_k^{ij} + 1 \\
		\vdots \\
		\aoi_k^{ij} + k_\tau^{ij} - 1
		\vphantom{\aoi_{k_\tau^{ij}-1|k}^{ij}}\\
		\aoi_k^{ij} + k_\tau^{ij} 
		\vphantom{\aoi_{k_\tau^{ij}-1|k}^{ij}}\\
		k_\tau^{ij} - k_f + 1 
		\vphantom{\aoi_{k_\tau^{ij}-1|k}^{ij}}\\
		\vdots
	\end{bmatrix}
	.
\end{salign}
The concept can be extended in a straightforward way for cases in which the threshold is transgressed more than once.

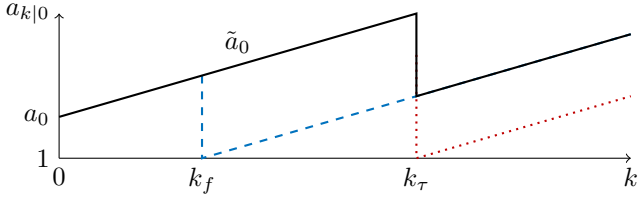
\begin{figure}[tb]
	\centering
	\begin{tikzpicture}[xscale=0.95, yscale = 0.55]
		\draw[->] (0,0) -- (8,0) node[anchor=north] {$k$};
		\draw[->] (0,0) -- (0,3.5) node[anchor=east] {$a_{k|0}$};
		\draw	(0,0) node[anchor=north] {$0$}
		(0,0) node[anchor=east] {$1$}
		(0,1) node[anchor=east] {$a_0$}
		(2,0) node[anchor=north] {$k_f$}
		(5,0) node[anchor=north] {$k_\tau$};
		
		\draw[dashed, thick, janniks_1A_blau] (2,2) -- (2,0) -- (8,3);
		\draw[dotted, thick, richards_1A_rot] (5, 2.5) -- (5,0) -- (8,1.5);
		\draw[thick] (0,1) -- (5,1+4*5/8) -- (5,1.5) -- (8,3);
		\draw (2.5,2.8) node {$\tr{a}_0$};
		
	\end{tikzpicture}
	\caption{Black line: generated forecast, given $k_f$ and $k_\tau$. Blue-dashed (red-dotted) line: AoI, if transmission would succeed in $k_f$ ($k_\tau$). Upper indices are omitted, current time-step is assumed to be $0$.
	}
	\label{fig::com::forecast}
\end{figure}

Due to the underlying stochastics, forecasts are of limited reliability. Decreasing the threshold $\tau$ naturally improves the reliability, however also reduces the amount of forecasts that can be detected within the fixed prediction horizon $\HR$. Depending on how the agents utilize the forecasts, a balance must be reached in which there are still enough reliable forecasts. This balance, together with an in-depth analysis of the network control performance under the presented control strategy will not be discussed in the scope of this paper. Instead, the focus lies on how the subsystem controllers make use of the forecasts to improve the overall control performance, as described in the next section.

\section{Distributed Control System}
\label{sec:controlsystem}

With respect to the distributed plant to be controlled, consider again a set $\mathcal{N}=\{1,\dots,\nc\}$ of discrete-time, linear time-invariant subsystems with the following dynamics:
\begin{salign}
	x_{k+1}\s{i}=A\s{i}x_k\s{i}+B\s{i}u_k\s{i},\quad i\in\mathcal{N}. \label{eq::state_equation}
\end{salign}
For the subsystem with index $i$, $x_k\s{i}$ denotes the state vector and $u_k\s{i}$ the input vector, both subject to polytopic constraints:
\begin{salign}
	&x_k\s{i}\in\mathcal{X}\s{i}=\setdef{x\in\mathbb{R}^{n_x\s{i}}}{C_x\s{i} \cdot x \leq b_x\s{i}},\\
	&u_k\s{i}\in\mathcal{U}\s{i}=\setdef{u\in\mathbb{R}^{n_u\s{i}}}{C_u\s{i} \cdot u \leq b_u\s{i}} \label{eq:input_constraint}.
\end{salign}
It is assumed that the subsystems are coupled through constraints and the control objective, but not through the dynamics \eqref{eq::state_equation} at this stage. Let the index-set $\Ni\subset\mathcal{N}$ contain the indices of all subsystems coupled to $i$ in this form, called \emph{neighbors}.

Since the exchange of information is limited to the communication network, each subsystem estimates the behavior of its neighbors based on the initially provided model of the neighbors' dynamics.
During system operation, each subsystem controller receives or holds possibly outdated information on the neighbors' states, and estimates their expected current state by applying forward time-shift to the neighbors' models. 
Thus, a local, augmented (and now dynamically coupled) model can be constructed for any subsystem:
\begin{align}
	\mathbf{x}_{k+l+1|k}\s{i} =&	\mathbf{A}\s{i} \mathbf{x}_{k+l|k}\s{i} + \mathbf{B}\s{i} u_{k+l|k}\s{i} \zu
	&+   \mathbf{B}_1\s{i} \mathbf{u}_{k+l|k}\s{i} + \mathbf{B}_2\s{i} \delta\mathbf{u}_{k+l|k}\s{i}.
	\label{eq::augmented_model}
\end{align}
In here, the augmented state vector is \mbox{$\mathbf{x}_k\s{i}=\begin{sbm}x_k\s{i}; \begin{sbm}x_k\s{ij}\end{sbm}_{j\in\Ni}\end{sbm}\in\mathbb{R}^{\mathbf{n}_x\s{i}}$}, the  augmented vector of nominal inputs of the neighbors is \mbox{$\mathbf{u}_k\s{i}=\begin{sbm}
		u_k\s{ij}\end{sbm}_{j\in\Ni}$}, and possible deviations of the neighbors' inputs from their nominal values are denoted by \mbox{$\delta\mathbf{u}_k\s{i}=\begin{sbm}\delta u_k\s{ij}\end{sbm}_{j\in\Ni}$}. Such a possible deviation of an input in time-step $k+l$ as predicted in time-step $k$  is defined by:
\begin{salign}
	u_{k+l}\s{j}=u_{k+l|k}\s{ij}+\delta u_{k+l|k}\s{ij}. \label{eq::def_deviation}
\end{salign}
The matrices of the model \eqref{eq::augmented_model} follow as 
described in \cite{hahn2018distributed}, e.g. if $\mathcal{N}\s{i}=j$:
\begin{align}
    \mathbf{A}\s{i}=\begin{sbm} A\s{i} &0\\ 0 & A\s{j}
    \end{sbm},
    \mathbf{B}\s{i}=\begin{sbm} B\s{i} \\ 0 
    \end{sbm},
    \mathbf{B}_1\s{i}=\mathbf{B}_2\s{i}=\begin{sbm}  0 \\B\s{j}
    \end{sbm}.
\end{align}


All subsystems $j\in\Ni$ communicate conservative uncertainty sets
$\Delta\mathcal{U}$ for their possible deviations to subsystem $i$, i.e. their selected inputs are guaranteed to be contained in these sets. Collecting all the communicated data, polytopic constraints for subsystem $i$ are written as:
\begin{equation}
	\label{eq:abcdef}
	\delta\mathbf{u}_{k+l|k}\s{i}\in\delta\mathbb{U}_{k+l|k}\s{i}=\mbox{prod}\left(\Delta\mathcal{U}_{k+l|k}\s{ij},j\in\Ni\right)
\end{equation}
where the sets $\Delta \mathcal{U}_{k+l|k}\s{ij}$ correspond to the communicated sets but shifted forward in time to the current time-step. Note, that also subsystem $i$ has to communicate such an uncertainty set to all its neighbors. The determination of this uncertainty set is described later in \eqref{eq:uncertaintyset}.

As already mentioned, the augmented model \eqref{eq::augmented_model} is dynamically coupled to neighbored subsystems such that coupling by state constraints can be formulated by:
\begin{salign}
	\mathbf{x}_{k+l|k}\s{i}\in\mathbb{X}\s{i}=\setdef{\mathbf{x}}{\mathbf{C}_x\s{i} \cdot\mathbf{x} \leq\mathbf{b}_x\s{i}}.
\end{salign}
Extending the model \eqref{eq::augmented_model} to the prediction model (c.f. \cite{hahn2019robust})
the behavior and constraints of the augmented system is predicted
up to the control horizon $\Hc$:
\begin{salign}
	\tbfi{x}_k=\tbfi{A}\mathbf{x}_{k|k}\s{i}+\tbfi{B}\tilde{u}_k\s{i}+&\tbfi{B}_1 \tbfi{u}_k+ \tbfi{B}_2 \delta\tbfi{u}_k \label{eq::pred_mod}
\end{salign}
with $\tilde{u}_k\s{i}$ as the input trajectory of subsystem $i$. The state trajectory $\tbfi{x}_k$ needs to satisfy
the state constraints:
\begin{align}
	&\tbfi{x}_k\in\tilde{\mathbb{X}}\s{i}=\mathbb{X}\s{i} \times \ldots \times \mathbb{X}_\xi\s{i}=\setdef{\tbf{x}}{\tbfi{C}_x\cdot \tbf{x} \leq \tbfi{b}_x} \label{eq::const_statetrajectory}
\end{align}
for all possible deviations of neighbored subsystems:
\begin{align}
	&\tbfi[\delta]{u}_k\in\delta\tilde{\mathbb{U}}_k\s{i}=\setdef{\tbf[\delta]{u}}{\tbfi{C}_\delta\cdot \tbf[\delta]{u} \leq \tbfi{b}_{\delta_k}},
\end{align}
where the set $\delta\tilde{\mathbb{U}}_k\s{i}$ is derived from \eqref{eq:abcdef}.

To compensate for these possible deviations from preceding neighbors, the control law of each subsystem is formulated as a disturbance feedback in stacked vector form:
\begin{salign}
	\tilde{u}_k\s{i}=\tilde{v}_k\s{i}+\tilde{\Delta}_k\s{i}=\tilde{v}_k\s{i}+K_{k}\s{i} \tbf[\delta]{u}_k\s{i} \label{eq::controllaw}
\end{salign}
with $\tilde{v}_k\s{i}$ as nominal input trajectory and $\tilde{\Delta}_k\s{i}$ as disturbance feedback term. Note that the $\tilde{\Delta}_k\s{i}$ vanishes if all neighbored subsystems do not deviate from communicated trajectories, and that the local input trajectory needs to satisfy a local time-varying constraint:
\begin{align}
	\tilde{u}_k\s{i}\in\tilde{\mathcal{U}}_k\s{i}=\setdef{\tilde{u}}{\tilde{C}_u\s{i}\cdot \tilde{u} \leq \tilde{b}_{u_k}\s{i}}\subseteq \mathcal{U}\s{i} \times \ldots \times \mathcal{U}\s{i}. \label{eq::const_inputtrajectory}
\end{align}
Even though \eqref{eq:input_constraint} is time-invariant, the time-variance of \eqref{eq::const_inputtrajectory} results from a self-updating mechanism to be formulated later in \eqref{eq::consistency}. Thus, the control trajectory to be computed in $k$ is constrained by the information about possible deviations from the nominal trajectory computed and communicated to connected subsystems in a previous time-step.

\begin{assume}\label{ass:relibaleForecast}
	In each time-step $k$, the network-controller provides reliable forecasts of the future age of information $\tilde{a}_k\s{ij}$ with \eqref{eq:forecasts} to each $j\in\Ni$ such that:
	\begin{enumerate}
		\item Predictions of the AoI for a time-step $k+l$ are upper bounds for the real AoI at time-step $k+l$, such that \mbox{$\aoi_{k+l}\s{ij}\leq \aoi_{k+l|k+r}\s{ij}\leq \aoi_{k+l|k}\s{ij}\; \forall k< r< l$} holds.
		\item If $\HR < \Hc$, then \mbox{$\aoi_{k+l|k}\s{ij}:=\bar{\aoi}\s{ij}\; \forall l \in\{\HR+1,\ \Hc \}$}, where $\bar{\aoi}\s{ij}=\max\limits_{k}\left(a_k\s{ij}\right)$, and $N$ is the horizon for the trajectory $\tr{\aoi}_k^{ij}$ as provided by the network controller.
		\item The disturbance feedback matrix $K_{k}\s{i}$ satisfies the structure:
		\begin{salign}
			K_{k}\s{i}=\begin{sbm}
				\mathbf{0} & \mathbf{0} & \cdots & \mathbf{0}\\
				\eta_{k+1,k-\bar\aoi_k\s{i}+1|k}\s{i} & \cdots & \cdots & \mathbf{0}\\
				\eta_{k+2,k-\bar\aoi_k\s{i}+1|k}\s{i} & \ddots 	& \cdots &\mathbf{0}\\
				\vdots & \ddots & \cdots & \mathbf{0}\\
				\eta_{k+\Hc-1,k-\bar\aoi_k\s{i}+1|k}\s{i} & \cdots & \eta_{k+\Hc-1,k+\Hc-2|k}\s{i} & \mathbf{0}
			\end{sbm}
		\end{salign}
		with
		\begin{salign}
			&\bar\aoi_k\s{i}=\max\limits_{j\in\Ni}\left(\aoi_k\s{ij}\right),\\
			&
			\eta_{n,m|k}\s{i}=\begin{sbm}\eta_{n,m|k}\s[\T]{ij}\end{sbm}^\T_{j\in\Ni}\in\mathbb{R}^{n_u\s{i} \times \mathbf{n}_u\s{i}}, \nonumber \\
			& \eta_{n,m|k}\s{ij}\in\begin{cases2}
				\mathbb{R} & \mbox{if } k-\aoi_k\s{ij}<m\leq n-\aoi_{n|k}\s{ij}\\
				\mathbf{0} & \mbox{else}.
			\end{cases2}.
		\end{salign}
	\end{enumerate}
\end{assume}

By substituting \eqref{eq::controllaw} in \eqref{eq::pred_mod}, the constraints \eqref{eq::const_statetrajectory} and \eqref{eq::const_inputtrajectory} are 
combined into the set of all admissible input trajectories  (c.f. \cite{hahn2018distributed,hahn2019robust}):
\begin{salign}
	\left(\tilde{v}_k\s{i},K_{k}\s{i}\right)\in \Pi_k\s{i}\left(\mathbf{x}_k\s{i}\right) \label{eq::admissible set}.
\end{salign}

From \cite{gross2014distributed}, the specification of local control goals with respect to the augmented state and input vectors is adopted, establishing a form of coupling through control objectives.



%
The local control objective with finite horizon $\Hc$ and stage cost $L\s{i}(x,u)=\norm{x}^2_\Qxi+ \norm{u}^2_\Qui$ is given by:
\begin{salign}
	J_k\s{i} = V_f\s{i}\left(\xi_{k+\Hc|k}\s{i}\right) + \sum_{l=0}^{\Hc-1}L\s{i}\left(\mathbf{x}_{k+l|k}\s{i},\begin{sbm}v_{k+l|k}\s[\T]{i}\; \mathbf{u}_{k+l|k}\s[\T]{i}\end{sbm}^\T\right)\quad \label{eq::cost_function}
\end{salign}
with $\Qxi=\QxiT>0,\Qui=\QuiT>0$, end cost function $V_f\s{i}(r)=\norm{r}^2_{P_\xi\s{i}}$ with terminal cost $P_\xi\s{i}=P_\xi\s[\T]{i}\geq 0$, and terminal state $\xi_k\s{i}$:
\begin{salign}
	\xi_{k+\Hc|k}\s{i}=\begin{sbm}\mathbf{x}_{k+\Hc-\bar\aoi\s{i}+1|k}\s[\T]{i},\dots,\mathbf{x}_{k+\Hc-1|k}\s[\T]{i},\mathbf{x}_{k+\Hc|k}\s[\T]{i}\end{sbm}^\T \label{eq::delayedsystem}
\end{salign}
In here, $\bar{\aoi}\s{i}=\max\limits_{k,j\in\Ni}\left(\aoi_k\s{ij}\right)$ is the maximally possible age of information.

\begin{assume}\label{ass::terminalset}
	Given a terminal set \mbox{$\mathbb{X}_\xi\s{i}\subseteq \mathbb{X}\s{i} \times \ldots \times \mathbb{X}\s{i}$} for the terminal state $\xi_k\s{i}$, it holds that:
	\begin{enumerate}
		\item The set $\mathbb{X}_\xi\s{i}$ is robust forward invariant with respect to the closed-loop system:
		\begin{salign}
			\xi_{k+1}\s{i}=A_\xi\s{i}\xi_k\s{i} +  B_\xi\s{i}  \mathbf{u}_k\s{i} \label{eq::terminalsystem}
		\end{salign}
		with
		\begin{salign}
			A_\xi\s{i}=\left(\begin{sbm}\mathbf{0} & I\\ \mathbf{0} & \mathbf{A}\s{i}\end{sbm}+\begin{sbm}\mathbf{0} & \mathbf{0}\\ \mathbf{B}\s{i}K\s{i} & \mathbf{0}\end{sbm}\right),\ B_\xi\s{i}=\begin{sbm}\mathbf{0}\\\mathbf{B}_1\s{i}\end{sbm}
		\end{salign}
		such that $\xi_{k+1}\s{i}\in\mathbb{X}_\xi\s{i}$ if $\xi_k\s{i}\in\mathbb{X}_\xi\s{i}$ for all $\mathbf{u}_k\s{i} \in prod\left(\mathcal{U}\s{j},j\in\Ni\right)$.
		\item The terminal state feedback $u_k\s{i}=K\s{i}\mathbf{x}_{k-\bar{\aoi}\s{i}+1}\s{i}=\begin{sbm} K\s{i} &\mathbf{0} \end{sbm}\xi_k\s{i} =K_\xi\s{i} \xi_k\s{i}$ satisfies the local input constraint $u_k\s{i}\in\mathcal{U}\s{i}$ for all $\xi_k\s{i}\in\mathbb{X}_\xi\s{i}$.
		\item The terminal cost ${P_\xi\s{i}}$ solves the Lyapunov equation 
		${A_\xi\s{i}}^\T P_\xi\s{i}A_\xi\s{i}-P_\xi\s{i}+Q_\xi\s{i}\preccurlyeq0$
		with \mbox{$Q_\xi\s{i}=diag\left(\begin{sbm}
				\mathbf{0}\\K\s{i}\end{sbm}^\T\Qui\begin{sbm}
				\mathbf{0}\\K\s{i}\end{sbm},\mathbf{0},\dots,\mathbf{0},\Qxi\right)$} and with respect to the autonomous and undisturbed closed-loop system: 
		\begin{align}
			\xi_{k+1}\s{i}=A_\xi \xi_k\s{i}. \label{eq:Lyap_CL}
		\end{align}
		\item $L(\mathbf{x}_{k+l|k}\s{i},\mathbf{0})=0$ for all $l\in\{\Hc-\aoi\s{i}+1,\dots,\Hc\}$ implies $\norm{\xi_{k+\Hc|k}\s{i}}^2_{P_\xi\s{i}}=0$.
	\end{enumerate}
\end{assume}

Minimizing the cost function \eqref{eq::cost_function} with respect to the set of admissible inputs \eqref{eq::admissible set}, the optimization problem to be solved in every time-step is the following:
\begin{salign}
	V\s{i}_k= &\min_{\left(\tilde{v}_k\s{i},K_{k}\s{i}\right)} 
	\norm{\xi_{k+\Hc|k}\s{i}}^2_{P_\xi\s{i}} \label{eq::QP} \zu
	& +\sum_{l=0}^{\Hc-1}\norm{\mathbf{x}_{k+l|k}\s{i}}^2_\Qxi+\norm{\begin{sbm}v_{k+l|k}\s[\T]{i}& \mathbf{u}_{k+l|k}\s[\T]{i}\end{sbm}^\T}^2_\Qui \nonumber\\
	\mbox{s.t.:} \quad & \mathbf{x}\s{i}_{k+l+1|k}=\mathbf{A}\s{i}\mathbf{x}_{k+l|k}\s{i}+\mathbf{B}\s{i} v_{k+l|k}\s{i} + \mathbf{B}_1\s{i}\mathbf{u}_{k+l|k}\s{i}, \nonumber\\
	& \text{and } \eqref{eq::admissible set}. \nonumber
\end{salign}

Solving \eqref{eq::QP} 
yields the nominal input trajectory $\tilde{v}_k\s{i}$ and the feedback matrix $K_{k}\s{i}$. Together with the bounds on $\tbfi[\delta]{u}_k$, the feeback matrix provides an upper bound for $\tilde{\Delta}_k\s{i}$, such that:
\begin{salign}
	\Delta\tilde{\mathcal{U}}_k\s{i}=\setdef{\tilde{\Delta}_k\s{i}}{\tilde{C}_u\s{i} \tilde{\Delta}_k\s{i}\leq \tilde{\gamma}_k\s{i}} \label{eq:uncertaintyset}
\end{salign}
is the uncertainty set of the nominal input trajectory of subsystem $i$ (for exact computation of $\tilde{\gamma}_k\s{i}$ c.f. \cite{hahn2018distributed,hahn2019robust}).

The current state, nominal input trajectory, and uncertainty set determine the behavior of subsystem $i$ for the next $\Hc$ time-steps, and are communicated to all neighboring
subsystems. Forcing subsystem $i$ 
in the next time-step $k+1$ to comply with the communicated information in time-step $k$, the time-varying input trajectory constraint \eqref{eq::const_inputtrajectory} needs to be shifted forward in time by one time-step:
\begin{salign}
	\tilde{b}_{u_{k+1}}\s{i}=\tilde{C}_u\s{i} \begin{sbm}v_{k+1|k}\s[\T]{i},\dots,v_{k+\Hc-1|k}\s[\T]{i}, \mathbf{0}\end{sbm}  \nonumber\zu
	+\begin{sbm}\gamma_{k+1|k}\s[\T]{i},\gamma_{k+\Hc-1|k}\s[\T]{i},\dots,b_u\s[\T]{i}\end{sbm} \label{eq::consistency}.
\end{salign}

\begin{theorem}
	If Asm. \ref{ass:relibaleForecast} and \ref{ass::terminalset}  hold, and \eqref{eq::QP} is feasible for all subsystems $i\in\mathcal{N}$ in time-step $k=0$, \eqref{eq::QP} remains feasible for all subsystems and $k>0$.
\end{theorem}


\begin{proof}
	If \eqref{eq::QP} is feasible in time-step $k$, it provides through \eqref{eq::admissible set} a nominal	input $v_{k|k}\s{i}$ which satisfies the time-varying input constraint $\mathcal{U}_k\s{i}$, and the global input constraint $\mathcal{U}\s{i}$ by definition of \eqref{eq::const_inputtrajectory}.
	Furthermore, $\tilde{u}_k\s{i}$ robustly steers the local system into the terminal constraint $\mathbb{X}_\xi\s{i}$ satisfying all local and global state constraints and time-varying input constraints for all possible deviations $\tbfi[\delta]{u}_k\in\delta\tilde{\mathbb{U}}_k\s{i}$.\\
	In all time-steps $k+l$ for $1\leq l <\Hc$, a feasible input $u_{k+l}\s{i}=u_{k+l|k}\s{i}$ is given by the tuple $\left(\tilde{v}_k\s{i},K_{k}\s{i}\right)$ designed in $k$ with:
	\begin{salign}
		u_{k+l}\s{i}:= v_{k+l|k}\s{i} + \sum_{j\in\mathcal{N}\s{i}}\;\sum_{m=k-\aoi_k\s{ij}+1}^{k+l-\aoi_{k+l|k}\s{ij}} \eta_{k+l,m|k}\s{ij} \delta u_{m|k}\s{ij}. \label{eq::controllaw2}
	\end{salign}
	Note, that the lower and upper bounds follow from Asm.  \ref{ass:relibaleForecast}.2, and that the last required deviation is $\delta u_{k+l-\aoi_{k+l|k}\s{ij}|k}\s{ij}$.
	Recall \eqref{eq::def_deviation} for time-step $k+l$, and the fact that the last exactly known input from subsystem $j$ in $k+l$ is given by $u_{k+l-\aoi_{k+l}\s{ij}}\s{j}$. Then, the last known deviation is $\delta u_{k+l-\aoi_{k+l}\s{ij}|k}\s{ij}$, which is at least as old as the desired one, if assumption \ref{ass:relibaleForecast}.1 holds with $\aoi_{k+l|k}\s{ij}\geq \aoi_{k+l}\s{ij}$.
	Through \eqref{eq::admissible set}, input $u_{k+l}\s{i}$ satisfies local input constraints, if all $\delta u_{m,k}\s{ij}\in\Delta\mathcal{U}_{m|k}\s{ij}$, which is true, if all subsystems $j$ update their local input trajectory constraint according to \eqref{eq::consistency}. Additionally, state constraints are also satisfied through \eqref{eq::admissible set}.
	
	In time-step $k+\Hc$, an admissible input can be calculated according to the terminal control law $u_{k+\Hc}\s{i}=K_\xi\s{i}\xi_{k+\Hc|k}\s{i}$, since the terminal state $\xi_{k+\Hc|k}\s{i}$ is robustly steered into $\mathbb{X}_\xi\s{i}$ through constraint \eqref{eq::admissible set} within \eqref{eq::const_statetrajectory}.
	If Asm. \ref{ass::terminalset} holds and $\mathbb{X}_\xi\s{i}$ is robustly invariant for any permissible $\mathbf{u}_{k+\Hc|k}\s{i}$, state and input constraints are satisfied.
	
	Thus, Asm. \ref{ass:relibaleForecast} and \ref{ass::terminalset} ensure the existence of the admissible input trajectory in $k+1$:
	\begin{salign}
		\tilde{u}_{k+1}\s{i}:=\begin{sbm}
			u_{k+1|k}\s[\T]{i}, \dots, u_{k+\Hc-1|k}\s[\T]{i}, (K_\xi\s{i}\xi_{k+\Hc|k}\s{i})^\T
		\end{sbm}^\T. \label{eq::u_traj_k+1}
	\end{salign}
	
	Feasibility of \eqref{eq::QP} for all times $k+l\; \forall l\in\mathbb{N}_{>0}$, follows directly by induction over k.
\end{proof}

According to the signal $\mathbf{u}_{k+l|k}\s{i}$ in the cost function \eqref{eq::cost_function}, stability of the distributed control system is proven with respect to the \textit{ISpS}-property, which is stated and proven below:
\begin{definition}\label{def::ISpS}
	From \cite{lazar2008input} it is well known, that an autonomous system \mbox{$z_{k+1}=f(z_k,w_k)$} with bounded disturbance \mbox{$w_k\in\mathcal{W}$} is \textit{input-to-state practical stable} (ISpS) in a forward positive invariant $\mathcal{Z}$, if there exist constants $d_1,d_2\geq0$, $a_1,a_2,a_3,a_e >0$, and $\mathcal{K}$-functions $\alpha_1(r)=a_1r^{a_e},\alpha_2(r)=a_2r^{a_e},\alpha_3(r)=a_3r^{a_e}$, and $\phi(r)$ respectively, and a function $V:\mathcal{Z}\rightarrow R_{\geq0}$ such that for all $z\in\mathcal{Z}$:
	\begin{salign}
		&\alpha_1\left(\norm{z_k}\right) \leq V\left(z_k\right) \leq \alpha_2\left(\norm{z_k}\right)+d_1, \label{eq::Lyapunovfunction1}\\
		& V\left(z_{k+1}\right)-V\left(z_k\right) \leq -\alpha_3\left(\norm{z_k}\right)+\phi\left(\norm{w_k}\right)+d_2 \label{eq::Lyapunovfunction2}
	\end{salign}
	hold for all $w_k \in\mathcal{W}$, where $\norm{z_k}$ denotes any norm. As in \cite{gross2016relaxed}, the relaxation of \eqref{eq::Lyapunovfunction2} to a multi-step definition is used with $\Hj\in\mathbb{N}_{\geq1}$, such that:
	\begin{salign}
		V\left(z_{k+\Hj}\right)-V\left(z_k\right)\leq \alpha_3\left(\norm{z_k}\right)+\Hj\phi\left(\norm{w_{[k,\Hj]}}\right)+\Hj d_2 \qquad \label{eq::Lyapunovfunction3}
	\end{salign}
	needs to be proven with $\norm{w_{[k,\Hj]}}=\max\limits_{r\in\{k,\dots,k+\Hj\}}\norm{w_r}$.
\end{definition}


Defining an extended state as the triple $z_k\s{i}=\left(\mathbf{x}_k\s{i},v_k\s{i},\mathbf{u}_k\s{i}\right)$, the norm and the trajectory are given by \mbox{$\norm{z_k\s{i}}^r_{Q_z\s{i}}:= \norm{\mathbf{x}_k\s{i}}_\Qxi^r+\norm{\begin{sbm}v_k\s{i};\mathbf{u}_k\s{i}\end{sbm}}_\Qui^r$}, and \mbox{$\tilde{z}_k\s{i}=\left(\tilde{\mathbf{x}}_k\s{i},\tilde{v}_k\s{i},\tilde{\mathbf{u}}_k\s{i}\right)$} respectively. Note that the cost function \eqref{eq::cost_function} explicitly depends on $\tilde{z}_k\s{i}$, but the value function \eqref{eq::QP} simply depends on the tuple $\left(\mathbf{x}_k\s{i},\tilde{\mathbf{u}}_k\s{i}\right)$, since $\tilde{\mathbf{x}}_k\s{i}$ follows from the auxiliary condition with initial state $\mathbf{x}_k\s{i}$, and $\tilde{v}_k\s{i}$ is an optimization variable. Thus, the condition $V_k\s{i}\left(\mathbf{x}_k\s{i},\tilde{\mathbf{u}}_k\s{i}\right)=J_k\s{i}\left(\tilde{\mathbf{x}}_k\s[\star]{i},\tilde{v}_k\s[\star]{i},\tilde{\mathbf{u}}_k\s{i}\right)=J_k\s{i}\left(\tilde{z}_k\s[\star]{i}\right)=V_k\s{i}$ holds, where $\star$ denotes the solution of the optimization \eqref{eq::QP}.



\begin{theorem}
	If Asm.s \ref{ass:relibaleForecast} and \ref{ass::terminalset} hold, system \eqref{eq::augmented_model} with control law \eqref{eq::controllaw} is ISpS, if the optimization problem \eqref{eq::QP} is feasible, and if weights $Q_x\s{i}$ and $Q_u\s{i}$ are chosen such that the set $\Omega\s{i}=\setdef{z}{L\s{i}(z)=0}$ is not empty.
\end{theorem}

\begin{proof}
	In order to establish the ISpS property, the set $\mathcal{Z}\s{i}$ is chosen equal to the set in the triple $\tilde{z}_k\s{i}$ for which \eqref{eq::QP} is feasible, and the set $\tilde{\Omega}\s{i}=\Omega\s{i}\times\dots\times\Omega\s{i}$ for the extended state trajectory.
	
	To consider deviations from previously communicated predictions of neighbored subsystems, i.e. $\mathbf{u}_{k+l|k}\s{i}\neq\mathbf{u}_{k+l|k+l}\s{i}$, the difference between two predictions is defined as \mbox{$\delta z_{k+l|k}\s{i}=z_{k+l|k+1}\s{i}-z_{k+l|k}\s{i}$}, and the difference between to predicted trajectories is $\delta\tilde{z}_k\s{i}=\begin{sbm} \delta z_{k+1|k}\s[\T]{i},\dots, \delta z_{k+\Hc|k}\s[\T]{i} \end{sbm}^\T$, with $\delta z_{k+\Hc|k}\s{i}=\left(\mathbf{x}_{k+\Hc|k+1}\s{i}-\mathbf{x}_{k+\Hc|k}\s{i},0,0\right)$.
	Since cost function \eqref{eq::cost_function} is a quadratic function, which is  zero if $\tilde{z}_k\s{i}\in\tilde{\Omega}\s{i}$ (according to Asm. \ref{ass::terminalset}.4), it is straightforward to see, that there exist lower and an upper comparison functions $\alpha_1\left(\norm{\tilde{z}_k\s{i}}_{\tilde{Q}_z\s{i}}\right)$ and $\alpha_2\left(\norm{\tilde{z}_k\s{i}}_{\tilde{Q}_z\s{i}}\right)$, such that \eqref{eq::Lyapunovfunction1} holds with $d_1=0$.
	
	To proof the reduction of cost \eqref{eq::Lyapunovfunction3}, consider the  optimal cost in $k$ if \eqref{eq::QP} is feasible:
	\begin{align}
		V_k\s{i}=&\norm{\xi_{k+\Hc|k}\s{i}}^2_{P_\xi\s{i}} + \norm{\mathbf{x}_{k|k}\s{i}}^2_\Qxi + \norm{\begin{sbm}v_{k|k}\s{i};\mathbf{u}_{k|k}\s{i}\end{sbm}}^2_\Qui  \zu
		&+\sum_{l=0}^{\Hc-2} \norm{\mathbf{x}_{k+l+1|k}\s{i}}^2_\Qxi +  \norm{\begin{sbm}v_{k+l+1|k}\s{i};\mathbf{u}_{k+l+1|k}\s{i}\end{sbm}}^2_\Qui
		\label{eq::opticostsk}.
	\end{align}
	Some general cost in $k+1$ are given by:
	\begin{align}
		J&_{k+1}\s{i}=\norm{\xi_{k+\Hc+1|k+1}\s{i}}^2_{P_\xi\s{i}} + \norm{\mathbf{x}_{k+\Hc|k+1}\s{i}}^2_\Qxi \zu &+\norm{\begin{sbm}v_{k+\Hc|k+1}\s{i};\mathbf{u}_{k+\Hc|k+1}\s{i}\end{sbm}}^2_\Qui  \zu
		&+\sum_{l=0}^{\Hc-2} \norm{\mathbf{x}_{k+l+1|k+1}\s{i}}^2_\Qxi +  \norm{\begin{sbm}v_{k+l+1|k+1}\s{i};\mathbf{u}_{k+l+1|k+1}\s{i}\end{sbm}}^2_\Qui
		\label{eq::generalcostsk}.
	\end{align}
	Now, first consider the case that there is no difference between predicted inputs of neighbored subsystems, i.e. \mbox{$\mathbf{u}_{k+l|k+1}\s{i}:=\mathbf{u}_{k+l|k}\s{i}$}. Applying in $k+1$ the input sequence defined in \eqref{eq::u_traj_k+1}, the result is a trajectory of admissible triples $\tilde{z}_{k+1}\s{i}$ without difference in predictions, i.e. $\delta \tilde{z}_{k}\s{i}=0$. An upper bound for the optimal cost difference is given by:
	\begin{salign}
		V_{k+1}&\s{i}-V_k\s{i}\leq
		\norm{\mathbf{x}_{k+\Hc|k+1}\s{i}}^2_\Qxi \zu
		&+\quad \norm{\begin{sbm}v_{k+\Hc|k+1}\s{i};\mathbf{u}_{k+\Hc|k+1}\s{i}\end{sbm}}^2_\Qui 
		+	\norm{\xi_{k+\Hc+1|k+1}\s{i}}^2_{P_\xi\s{i}} \zu
		&- \quad \norm{\mathbf{x}_{k|k}\s{i}}^2_\Qxi -
		\norm{\begin{sbm}v_{k|k}\s{i}; \mathbf{u}_{k|k}\s{i}\end{sbm}}^2_\Qui-\norm{\xi_{k+\Hc|k}\s{i}}^2_{P_\xi}.
	\end{salign}
	Note the following three facts:
	\begin{enumerate}
		\item \mbox{$\mathbf{x}_{k+\Hc|k+1}\s{i}:=\mathbf{x}_{k+\Hc|k}\s{i}$}, and
		\mbox{$v_{k+\Hc|k+1}\s{i}:=K_\xi\s{i}\xi_{k+\Hc|k}\s{i}$},
		such that: $\norm{\begin{sbm}K_\xi\s{i} \xi_{k+\Hc|k}\s{i};\mathbf{u}_{k+\Hc|k+1}\s{i} \end{sbm}}^2_\Qui \leq \norm{\begin{sbm}K_\xi\s{i} \xi_{k+\Hc|k}\s{i};\mathbf{0} \end{sbm}}^2_\Qui + \norm{\begin{sbm}\mathbf{0};\mathbf{u}_{k+\Hc|k+1}\s{i} \end{sbm}}^2_\Qui$ and
		$\norm{\mathbf{x}_{k+\Hc|k}\s{i}}^2_\Qxi+\norm{\begin{sbm}K_\xi\s{i} \xi_{k+\Hc|k}\s{i};\mathbf{0} \end{sbm}}^2_\Qui=\norm{\xi_{k+\Hc|k}\s{i}}^2_{Q_\xi\s{i}}$.
		\item $\xi_{k+\Hc+1|k+1}\s{i}=A_\xi\s{i}\xi_{k+\Hc|k+1}\s{i}+B_\xi\s{i}\mathbf{u}_{k+\Hc|k+1}$ with $\xi_{k+\Hc|k+1}\s{i}:=\xi_{k+\Hc|k}\s{i}$ such that:
		$\norm{\xi_{k+\Hc+1|k+1}\s{i}}^2_{P_\xi\s{i}} \leq \norm{A_\xi\s{i} \xi_{k+\Hc|k}\s{i} }^2_{P_\xi\s{i}}+\norm{B_\xi\s{i} \mathbf{u}_{k+\Hc|k+1}\s{i}}^2_{P_\xi\s{i}}$.        
		\item $P_\xi\s{i}$ is assumed to solve the Lyapunov equation for the closed-loop system \eqref{eq::terminalsystem}, such that: \mbox{$\norm{A_\xi\s{i} \xi_{k+\Hc|k}\s{i}}^2_{P_\xi\s{i}}-\norm{\xi_{k+\Hc|k}\s{i}}^2_{P_\xi\s{i}}+\norm{\xi_{k+\Hc|k}\s{i}}^2_{Q_\xi\s{i}}\leq 0$}.
	\end{enumerate}
	Considering these facts, 
	an upper bound results according to Asm. \ref{ass::terminalset}:
	\begin{salign}
		V_{k+1}\s{i}-V_{k}\s{i}\leq &-\norm{\mathbf{x}_{k|k}\s{i}}^2_\Qxi-
		\norm{\begin{sbm}v_{k|k}\s{i};\mathbf{u}_{k|k}\s{i}\end{sbm}}^2_\Qui  \zu 
		&+\norm{\begin{sbm}\mathbf{0};\mathbf{u}_{k+\Hc|k+1}\s{i}\end{sbm}}^2_\Qui + \norm{B_\xi\s{i}\mathbf{u}_{k+\Hc|k+1}\s{i}}^2_{P_\xi\s{i}}. \label{eq::decrease_nd}
	\end{salign}
	Since $\mathbf{u}_{k+\Hc|k+1}\s{i}\in\mbox{prod}\left(\mathcal{U}\s{j},j\in\Ni\right)$ is bounded, there exists a constant $d_2$ and a comparison function $\alpha_3$ such that:
	\begin{salign}
		V_{k+1}\s{i}-V_{k}\s{i}\leq -\alpha_3\left(\norm{z_k\s{i}}_{Q_z\s{i}}\right)+d_2 \label{eq::upperbound1}.
	\end{salign}
	Now consider the general case with possible deviations of previously communicated predictions, i.e., $\delta \tilde{z}_{k}\s{i}\neq0$. Analogously to $\delta \tilde{z}_{k}\s{i}$, the difference of two predicted  terminal states is denoted by $\delta\xi_{k+l|k}\s{i}$. 
	To the norm of possible deviations from above, the triangle inequality applied for all $l\in\{0,\dots,\Hc-2\}$ leads to:
	\begin{salign}
		\norm{z_{k+l+1|k+1}\s{i}}-\norm{z_{k+l+1|k}\s{i}}\leq\norm{\delta z_{k+l+1|k}\s{i}},
	\end{salign}
	and to the state in fact 1):
	\begin{salign}
		\norm{\mathbf{x}_{k+\Hc|k+1}\s{i}}^2_\Qxi \leq\norm{\mathbf{x}_{k+\Hc|k}\s{i}}^2_\Qxi+ \norm{\delta z_{k+\Hc|k}\s{i}}^2_{Q_z\s{i}},
	\end{salign}
	and to the input in fact 1):
	\begin{salign}
		\norm{\begin{sbm}K_\xi\s{i}\xi_{k+\Hc|k+1}\s{i};\mathbf{0}\end{sbm}}_{Q_u\s{i}} \leq &\norm{\begin{sbm}K_\xi\s{i}\xi_{k+\Hc|k}\s{i};\mathbf{0} \end{sbm}}_{Q_u\s{i}}\zu &+\norm{\begin{sbm}K_\xi\s{i}\delta\xi_{k+\Hc|k}\s{i};\mathbf{0} \end{sbm}}_{Q_u\s{i}},
	\end{salign}
	and to the terminal state in fact 2):
	\begin{salign}
		\norm{A_\xi\s{i}\xi_{k+\Hc|k+1}\s{i}}^2_{P_\xi\s{i}}\leq\norm{A_\xi\s{i}\xi_{k+\Hc|k}\s{i}}^2_{P_\xi\s{i}}+ \norm{A_\xi\s{i}\delta \xi_{k+\Hc|k}\s{i}}^2_{P_\xi\s{i}}.
	\end{salign}
	Then, the upper bound \eqref{eq::upperbound1} is obtained to:
	\begin{salign}
		V_{k+1}\s{i}-V_{k}\s{i}\leq& -\alpha_3\left(\norm{z_k\s{i}}_{Q_z\s{i}}\right)+d_2+\norm{\begin{sbm}K_\xi\s{i}\delta\xi_{k+\Hc|k}\s{i};0 \end{sbm}}_{Q_u\s{i}} \zu
		& + \norm{A_\xi\s{i}\delta \xi_{k+\Hc|k}\s{i}}^2_{P_\xi\s{i}} + 
		\sum_{l=1}^{\Hc} \norm{\delta {z}_{k+l|k}\s{i}}^2_{Q_z\s{i}}.
	\end{salign}
	Since a difference in the terminal state $\delta \xi_{k+\Hc|k}\s{i}$ comprises multiple differences in the augmented states $\mathbf{x}_{k+l|k}\s{i}$, it strictly depends on the difference of predictions $\delta\tilde{z}_{k}$. Therefore, it is straightforward that there exists a comparison function $\phi$, such that:
	\begin{salign}
		V_{k+1}\s{i}-V_k\s{i}\leq -\alpha_3\left(\norm{z_k\s{i}}_{Q_z\s{i}}\right)+ \phi\left(\norm{\delta\tilde{z}_k\s{i}}_{\tilde{Q}_z\s{i}}\right) +d_2
	\end{salign}
	holds.
	
	Since there is no guarantee to receive new information in $k+1$, and possible deviations of previously communicated predictions are not known before new data arrives, the scheme needs to by applied recursively. Assuming $\Hj=\Hc\geq \bar a\s{i}$, data will be known and the decrease of the value function is upper bounded by:
	\begin{salign}
		\begin{aligned}
			V_{k+\Hc}\s{i}-V_k\s{i}\leq& -\alpha_3\left(\norm{\tilde{z}_k\s{i}}_{\tilde{Q}_z\s{i}}\right)\zu
			&+\Hc\sigma\left(\norm{\delta\tilde{z}_{[k:k+\Hc]}\s{i}}_{\tilde{Q}_z\s{i}}\right) +\Hc d_2,
		\end{aligned}
		\label{eq::upperboundend}
	\end{salign}
	which complies to \eqref{eq::Lyapunovfunction3} and implies the ISpS property.
\end{proof}

\section{Simulation Example}
\label{sec:simulation}

In this section, the presented framework is applied to the example of 3 autonomous vehicles (subsystems denoted by S1 to S3 hereafter) which move as a platoon. The common control objective consists of minimizing the distances between the vehicles, while avoiding collision. It is assumed that the vehicles communicate over wireless channels, and that the tasks of the centralized network controller are either performed by one of the vehicles itself, or by a road side unit. Either way, the platoon is organized in such a way that each subsystem (corresponding to an agents in the \CNET) only requires data from its immediate predecessor. Hence, S1 acts autonomously, S2 gets data from S1 (with AoI $\aoi^{2,1}_k$), and S3 receives data from S2 (with AoI $\aoi^{3,2}_k$), as illustrated in \figref{fig::com::setup}.

\begin{figure}[tb]
	\centering
	\begingroup%
	\makeatletter%
	\newcommand*\fsize{\dimexpr\f@size pt\relax}%
	\newcommand*\lineheight[1]{\fontsize{\fsize}{#1\fsize}\selectfont}%
	\ifx\svgwidth\undefined%
	\setlength{\unitlength}{226.77165354bp}%
	\ifx\svgscale\undefined%
	\relax%
	\else%
	\setlength{\unitlength}{\unitlength * \real{\svgscale}}%
	\fi%
	\else%
	\setlength{\unitlength}{\svgwidth}%
	\fi%
	\global\let\svgwidth\undefined%
	\global\let\svgscale\undefined%
	\makeatother%
	\begin{picture}(1,0.16915594)%
		\lineheight{1}%
		\setlength\tabcolsep{0pt}%
		\put(0,0){\includegraphics[width=\unitlength,page=1]{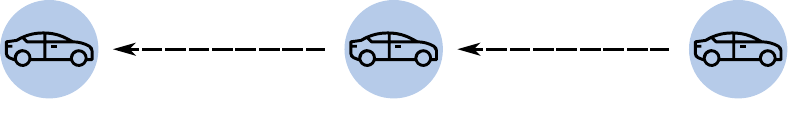}}%
		\put(0.28146616,0.06665409){\color[rgb]{0,0,0}\makebox(0,0)[t]{\lineheight{1.25}\smash{\begin{tabular}[t]{c}
						data
		\end{tabular}}}}%
		\put(0.71896621,0.06665409){\color[rgb]{0,0,0}\makebox(0,0)[t]{\lineheight{1.25}\smash{\begin{tabular}[t]{c}
						data
		\end{tabular}}}}%
		\put(0.93874189,0.00063501){\color[rgb]{0,0,0}\makebox(0,0)[t]{\lineheight{1.25}\smash{\begin{tabular}[t]{c}
						S1
		\end{tabular}}}}%
		\put(0.50124189,0.00063501){\color[rgb]{0,0,0}\makebox(0,0)[t]{\lineheight{1.25}\smash{\begin{tabular}[t]{c}
						S2 / $\aoi^{2,1}_k$
		\end{tabular}}}}%
		\put(0.06374189,0.00063501){\color[rgb]{0,0,0}\makebox(0,0)[t]{\lineheight{1.25}\smash{\begin{tabular}[t]{c}
						S3 / $\aoi^{3,2}_k$
		\end{tabular}}}}%
	\end{picture}%
	\endgroup%
	\vspace{3mm}
	\caption{Structure of the platoon example.}
	\label{fig::com::setup}
\end{figure}
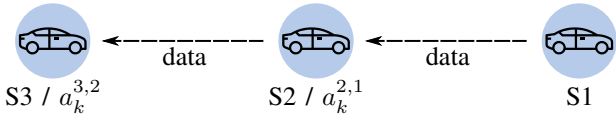

Given the goal of \textit{robust} control of the platoon in the \CSYS, the forecasts predicted by the network controller in the \CNET\ have to be reliable according to assumption \ref{ass:relibaleForecast}.1.
This is accomplished by 2 minor tweaks to the \CNET\ model: First, the DTMC $\{s_k\}$ (which governs the quality of the links) is modeled with deterministic transition probabilities. In particular, each state of $\{s_t\}$ dictates the two mean transition probabilities $\e{p}_k^2$ and $\e{p}_k^3$ for the transmission between S1$\,\to\,$S2 and S2$\,\to\,$S3, respectively, as illustrated in \figref{fig::DTMC}.

Secondly, the following minor deviation from the MPC paradigm for the network controller is considered: Instead of optimizing for the entire horizon without regard of the previously calculated optimal control trajectory, the network controller only optimizes for the last step of the horizon, while all prior controls are given by the previous trajectory. Though this constrains the control of the \CNET, it enables the strict robustness result for the control of the \CSYS\ from the previous section.

\begin{figure}
	\centering
	\small
	\begin{tikzpicture}[xscale = 0.1, yscale = -0.1]
		\tikzset{state/.style={minimum width=12mm, draw, circle,
				fill=richards_1A_grun}};
		
		\node[state, minimum width=16mm] (s1) at (0,0) {};
		\node at (0,0) {\shortstack{ 
				$\e{q}_{(1)}^1 = 1$ 
				\\[-1mm] 
				$\e{q}_{(1)}^2 = 1$ }};
		\node[state, minimum width=16mm] (s2) at (22,0) {};
		\node at (22,0) {\shortstack{
				\thickmuskip=1mu $\e{q}_{(2)}^1 =0.85$ 
				\\[-1mm] 
				\thickmuskip=1mu$\e{q}_{(2)}^2 =0.85$
		}};
		\node[state] (s3) at (44,0) {\shortstack{ $0.85$ \\[1mm] $1$ }};
		\node[state] (s4) at (66,0) {\shortstack{ $0.85$ \\[1mm] $0.85$ }};
		\node[state] (s5) at (66,20) {\shortstack{ $1$ \\[1mm] $0.85$ }};
		\node[state] (s6) at (44,20) {\shortstack{ $0.85$ \\[1mm] $0.85$ }};
		\node[state] (s7) at (22,20) {\shortstack{ $1$ \\[1mm] $1$ }};
		\node[state] (s8) at (0,20) {\shortstack{ $1$ \\[1mm] $1$ }};

		\draw[every loop]
		(s1) edge[auto] node {$1$} (s2)
		(s2) edge[auto] node {$1$} (s3)
		(s3) edge[auto] node {$1$} (s4)
		(s4) edge[auto] node {$1$} (s5)
		(s5) edge[auto] node {$1$} (s6)
		(s6) edge[auto] node {$1$} (s7)
		(s7) edge[auto] node {$1$} (s8)
		(s8) edge[auto] node {$1$} (s1);
		
		\node[draw, rectangle, anchor = west] (ini) at (10,11) {initial};    
		
		\draw (ini) edge (s1);
	\end{tikzpicture}
	\caption{DTMC $\{s_k\}$ governing the values of the mean transmission success probabilities by assigning $\e{p}_k^1 = \e{q}_{(s_k)}^1$ (top entry) and $\e{p}_k^2 = \e{q}_{(s_k)}^2$ (bottom entry).}
	\label{fig::DTMC}
\end{figure}
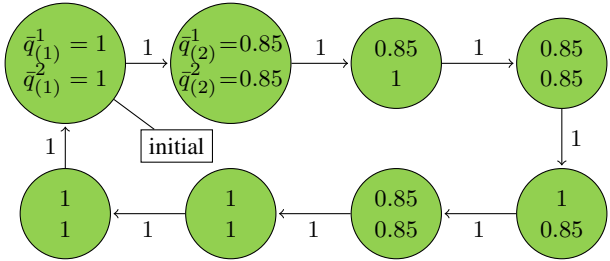

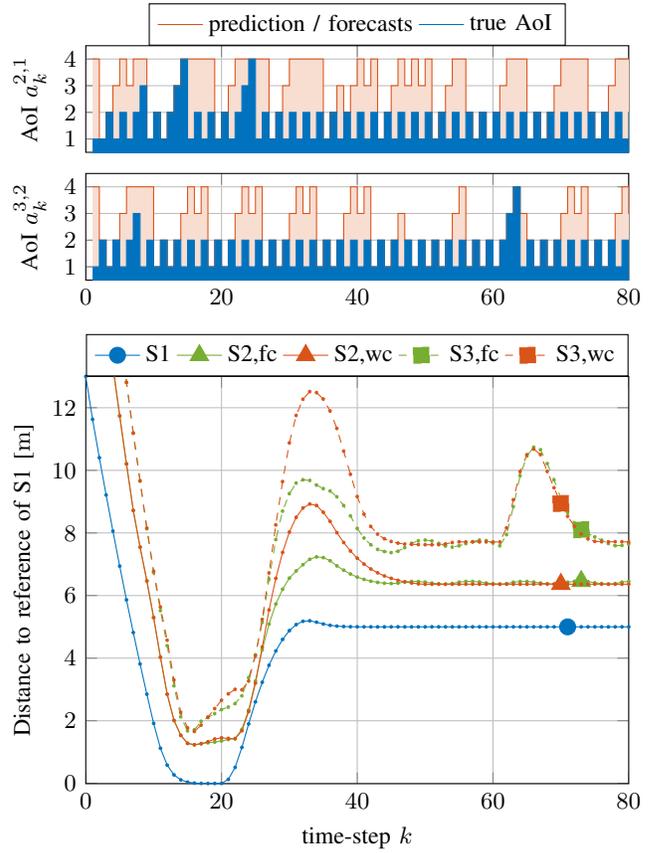
\begin{figure}
	\pgfplotsset{compat=1.3}
	\begin{tikzpicture}[x=1cm,y=1cm]
		\pgfplotsset{
			every axis legend/.append style={
				at={(0.5,1)},
				anchor=south
		}}
		
		\begin{axis}[
			font=\small,
			at={(0cm, 0mm)},
			anchor={north west},
			width = 88mm,
			height = 30 mm,
			xticklabels={,,},
			ylabel=AoI $\aoi^{2,1}_k$,
			xmin = 0, xmax = 80,
			ymin = 0.5, ymax = 4.5,
			grid = major,
			legend columns = 2
			]

			\addplot[name path=B,const plot,color=janniks_1A_rot] table[x expr=(\thisrow{t}+0), y expr=(\thisrow{pred1}+0)] {simdata/richard_data.txt};
			\addplot[name path=A,const plot,color=janniks_1A_blau]  table[x=t, y=aoi1] {simdata/richard_data.txt};
			\addplot[name path=C,const plot,draw=none]  coordinates{(1,0)(80,0)};
			
			\addplot[janniks_1A_blau!100] fill between[of=A and C];
			\addplot[janniks_1A_rot!20] fill between[of=A and B];
			\legend{ {prediction / forecasts},{true AoI} }
		\end{axis}
		
		\begin{axis}[
			font=\small,
			at={(0cm, -17mm)},
			anchor={north west},
			width = 88mm,
			height = 30 mm,
			ylabel=AoI $\aoi^{3,2}_k$,
			xmin = 0, xmax = 80,
			ymin = 0.5, ymax = 4.5,
			grid = major,
			legend columns = 3
			]
			
			\addplot[name path=B,const plot,color=janniks_1A_rot] table[x expr=(\thisrow{t}+0), y expr=(\thisrow{pred2}+0)] {simdata/richard_data.txt};
			\addplot[name path=A,const plot,color=janniks_1A_blau]  table[x=t, y=aoi2] {simdata/richard_data.txt};
			\addplot[name path=C,const plot,draw=none]  coordinates{(1,0)(80,0)};
			
			\addplot[janniks_1A_blau!100] fill between[of=A and C];
			\addplot[janniks_1A_rot!20] fill between[of=A and B];
		\end{axis}

		\begin{axis}[
			font=\small,
			at={(0cm,-44mm)},
			anchor={north west},
			width = 88mm,
			height = 70mm,
			xlabel={time-step $k$},
			ylabel={Distance to reference of S1 [\SI{}{\meter}]},
			xmin = 0, xmax = 80,
			ymin = -0.01, ymax = 13,
			grid = major,
			legend columns = 5
			]
			\addplot[color=janniks_1A_blau,mark=*,mark size=0.5pt]  table[x=t, y=y1] {simdata/jannik_data.txt};
			\addplot[color=janniks_1A_blau,mark=*,mark size=3pt, restrict x to domain = 71:71]  table[x=t, y=y1] {simdata/jannik_data.txt};
			
			
			\addplot[color=janniks_1A_grun,mark=*,mark size=0.5pt] table[x=t, y=y2_PNC] {simdata/jannik_data.txt};
			\addplot[color=janniks_1A_grun,mark=triangle*,mark size=4pt,restrict x to domain=73:73] table[x=t, y=y2_PNC] {simdata/jannik_data.txt};
			\addplot[color=janniks_1A_rot,mark=*,mark size=0.5pt] table[x=t, y=y2_WC] {simdata/jannik_data.txt};
			\addplot[color=janniks_1A_rot,mark=triangle*,mark size=4pt, restrict x to domain=70:70] table[x=t, y=y2_WC] {simdata/jannik_data.txt};
			
			
			\addplot[color=janniks_1A_grun,mark=*,mark size=0.5pt, dashed] table[x=t, y=y3_PNC] {simdata/jannik_data.txt};
			\addplot[color=janniks_1A_grun,mark=square*,mark size=03pt,restrict x to domain=73:73, dashed] table[x=t, y=y3_PNC] {simdata/jannik_data.txt};
			\addplot[color=janniks_1A_rot,mark=*,mark size=0.5pt, dashed] table[x=t, y=y3_WC] {simdata/jannik_data.txt};
			\addplot[color=janniks_1A_rot,mark=square*,mark size=03pt,restrict x to domain=70:70, dashed] table[x=t, y=y3_WC] {simdata/jannik_data.txt};
			
			\legend{ , S1, , {S2,fc}, ,  {S2,wc} ,, {S3,fc},, {S3,wc}}
			
		\end{axis}		
	\end{tikzpicture}
	\caption{Predicted and actual values for the AoI (top), and distance of vehicles to the reference of S1 with unpredicted setpoint change of S1 at $k=20$ with use of the forecasts (fc) in constrast to the use of a worst case delay (wc) (bottom).} 
	\label{fig::simulation}
\end{figure}
Doing so, the evolution of the AoI is shown in \figref{fig::simulation} (top) for a random realization of the link transmissions variables $p_k^1$ and $p_k^2$. The blue-shaded area represents the actual AoI $\bar\aoi\s{2}_k$ and $\bar\aoi\s{3}_k$, that is at least 1 and at most 4. Additionally, the red-shaded area represents the predicted AoI. Not quite visible is the fact that the red-shaded bars always include the blue ones.

The local dynamics \eqref{eq::state_equation} of the subsystems S1, S2, and S3 are chosen here to be parameterized identically with:
\begin{align}
	A\s{i}=\begin{sbm}
		1 & 0.3\\
		0 & 1
	\end{sbm},\
	B\s{i}=\begin{sbm}
		0.045\\
		0.3
	\end{sbm},\ i\in\mathcal{N}=\{1,2,3\}.
\end{align}
The state $x_k\s{i}=\begin{sbm}x_{1,k}\s{i} & x_{2,k}\s{i}\end{sbm}^\T$ consists of the position $x_{1,k}\s{i}$ and velocity $x_{2,k}\s{i}$. The input is given by the acceleration $u_k\s{i}$, 
constraint to  $\norm{u_k\s{1}}\leq 1.98$, $\norm{u_k\s{2}}\leq3$, and \mbox{$\norm{u_k\s{3}}\leq5$}.
The first vehicle S1 follows an internal reference, such that the focus is on the behavior of S2 and S3, with augmented states \mbox{$\mathbf{x}_k\s[\T]{i}=\begin{sbm}x_k\s[\T]{i} & x_k\s[\T]{i-1}\end{sbm}, i\in\{2,3\}$}. The weights of the cost functions are chosen to \mbox{$\Qxi=\begin{sbm}
		Q & -Q\\ -Q & Q
	\end{sbm}$}, with \mbox{$Q=\diag{}{5,1}$}, and \mbox{$\Qui=\begin{sbm}
		0.1 & -0.1\\
		-0.1 & 0.1
	\end{sbm}$} respectively. 
Note, that $Q_x\s{i}$ and $Q_u\s{i}$ are indefinite, but a decrease of stage cost in direction of all local states and inputs is still guaranteed.
To avoid collisions, the augmented state is constrained to $0\leq x_{1,k}\s{i-1}-x_{1,k}\s{i} \leq 200$. According to \figref{fig::simulation} (top), the maximum possible AoI is $4$. Hence, the terminal control laws need to compensate $4$ time-steps and are chosen by \mbox{$u_k\s{2}=\begin{sbm} -0.03& -0.54& 0.03& 0.54 \end{sbm}\textbf{x}_{k-4}\s{2}$}, and \mbox{$u_k\s{3}=\begin{sbm} -0.06& -0.6& 0.06& 0.6 \end{sbm}\textbf{x}_{k-4}\s{3}$}.
With a control horizon of \mbox{$\Hc=8$}, the simulation starts at $k=0$ with initial states $x_0\s{ref}=[0\ 5]^\T$, $x_{0}\s{1}=[-13\ 5]^\T$, $x_{0}\s{2}=[-20\  5]^\T$, and $x_{0}\s{3}=[-25\  5]^\T$, where $ref$ denotes the internal reference of S1. This reference corresponds to accelerating with a constant value of $u_k\s{ref}=1$. The vehicles S2 and S3 minimize the distance to their predecessor, and S1 to its reference respectively.

In $k=20$, 
an unexpected change of the setpoint of S1 is modeled to keep a constant distance of 5 meters to the reference. Thus, S1 suddenly decelerates and deviates from the prediction, which was communicated to S2 before (but remains inside the bounds predicted at $k=19$). As soon as S2 receives the information of this deviation, it reacts with deviating from its own plan previously predicted and communicated to S3 (but again within the communicated bounds). This scheme repeats for S3. Figure~\ref{fig::simulation} shows the distance of all subsystems with respect to the reference of S1. The control results differ with respect to the use of the forecasted age of information (fc) compared to the use of the worst case age (wc). The trajectory of the last subsystem is shown as dashed line, the trajectory of the first subsystem as a blue solid line.
At $k=60$, a significant delay of sending information from S2 to S3 in the communication network is modeled (see \figref{fig::simulation} top), such that S3 has to decelerate and to keep more distance to S2. This  results in the same behavior according to the use of the worst case age of information. Prior delays in the communication network have lower effect caused by general transient effects of the distributed plant.

The figure shows that the maximum distance of the platoon can be reduced significantly for S2 and S3 by using the forecasts within the predictive vehicle controllers. For the considered simulation, the distance of the platoon (measured from S1 to S3) caused by the setpoint change of S1 is reduced by over $37\%$, from \SI{7.5}{\meter} to \SI{4.7}{\meter}, still avoiding collisions in a robust way. This improvement of performance mainly results from the reduced predicted age of information, and the less restrictive constraints constructed thereon in the subsystem controllers.

\section{Conclusions}
\label{sec:conclusion}

A new control scheme for a class of cyber-physical system has been presented, which interweaves the control strategy of the communication system with the control strategy of the distributed plant. As a main result, robust stability in the sense of ISpS has been proven for the MPC scheme of the distributed plant, given the assumption that the predictive network controller provides  reliable forecasts of the age of information. 

The principle of controlling the communication network by a predictive control law offers improvements in two further respects: First, foreseen variations of the communication network can be used to decrease the communication delay. Secondly, the incorporation of future communication delays (which are a by-product of the predictive network control) decreases the conservativeness of the robust DMPC scheme for plant control. The amount of improvement is difficult to quantify in a general sense, but strongly depends on the difference of predicted delay to the upper bound of possible delays. 
The presented methodology is relevant for all applications for which a controlled distributed plant interacts through a communication network with time-varying communication delay. With respect to 5G communication, this applies to many cyber-physical systems, including all autonomous applications of driving and robotics.

Future work will extend the considerations to a framework in which the distributed plant control will consider the age of information in stochastic representation.  
\bibliographystyle{IEEEtran}
\bibliography{references_short}

\newpage

This work has been submitted to the IEEE for possible publication. Copyright may be transferred without notice, after which this version may no longer be accessible.

\end{document}